\documentclass[sigconf,9pt]{acmart}

\acmYear{Hongshi Tan, Xinyu Chen, Yao Chen, Bingsheng He, Weng-Fai Wong. 2021}\copyrightyear{2021}
\setcopyright{rightsretained}
\acmConference[ICS '21]{2021 International Conference on Supercomputing}{June 14--17, 2021}{Virtual Event, USA}
\acmBooktitle{2021 International Conference on Supercomputing (ICS '21), June 14--17, 2021, Virtual Event, USA}
\acmPrice{}
\acmDOI{10.1145/3447818.3461664}
\acmISBN{978-1-4503-8335-6/21/06}

\usepackage{bm}
\usepackage{cleveref}
\usepackage{pifont}
\usepackage{booktabs}
\usepackage{tabularx}
\usepackage{graphicx}
\usepackage{amsthm}
\usepackage{amsmath,amsfonts}
\usepackage{algorithm}
\usepackage{algpseudocode}
\usepackage{algorithmicx}
\usepackage{graphicx}
\usepackage{comment}
\usepackage{multirow}
\usepackage{subcaption}
\usepackage{comment}
\usepackage{dcolumn}

\usepackage{bm}
\usepackage{mathtools}

\usepackage{pifont}
\usepackage{multicol,threeparttable, tablefootnote}

\usepackage[skip=2pt]{caption} 

\usepackage{lipsum}
\usepackage{xargs}
\usepackage[colorinlistoftodos,prependcaption]{todonotes}
\newcommandx{\unsure}[2][1=]{\todo[linecolor=red,backgroundcolor=red!25,bordercolor=red,#1]{#2}}
\newcommandx{\change}[2][1=]{\todo[linecolor=blue,backgroundcolor=blue!25,bordercolor=blue,#1]{#2}}
\newcommandx{\info}[2][1=]{\todo[linecolor=OliveGreen,backgroundcolor=OliveGreen!25,bordercolor=OliveGreen,#1]{#2}}
\newcommandx{\improvement}[2][1=]{\todo[linecolor=Plum,backgroundcolor=Plum!25,bordercolor=Plum,#1]{#2}}
\newcommandx{\thiswillnotshow}[2][1=]{\todo[disable,#1]{#2}}
\usepackage{overpic}
\newcolumntype{d}[1]{D{.}{.}{#1}}

\usepackage[utf8]{inputenc}
\usepackage{listings}
\usepackage{fixltx2e}
\usepackage{breqn}
\usepackage{soul}
\usepackage{balance}

\usepackage[inline]{enumitem}

\usepackage[normalem]{ulem}               
\newcommand\redout{\bgroup\markoverwith
{\textcolor{red}{\rule[0.5ex]{2pt}{0.8pt}}}\ULon}

\lstdefinestyle{mystyle}{
  backgroundcolor=\color{white},
  xleftmargin = 0.2cm,
  commentstyle=\color{codegreen},
  keywordstyle=\color{magenta},
  numberstyle=\footnotesize\color{codegray},
  stringstyle=\color{codepurple},
  basicstyle=\ttfamily\footnotesize,
  breakatwhitespace=false,
  breaklines=true,
  captionpos=b,
  keepspaces=false,
  numbers=left,
  numbersep=5pt,
  showspaces=false,
  showstringspaces=false,
  showtabs=false,
  tabsize=2,
  lineskip=-2ex 
}
\usepackage{tikz}
\usetikzlibrary{decorations.pathreplacing,calc,shapes,positioning,tikzmark}
\usepackage{overpic}

\begin{document}
\fancyhead{}
\affiliation{\LARGE Hongshi Tan$^{1}$, Xinyu Chen$^{1}$, Yao Chen$^{2}$, Bingsheng He$^{1}$, Weng-Fai Wong$^{1}$ \vspace{1mm}} 
\affiliation{\large
$^{1}$School of Computing, National University of Singapore\hspace{5mm}
$^{2}$Advanced Digital Sciences Center, Singapore
\vspace{1mm}}

\title[ThundeRiNG: Random Number Generation on FPGAs]{ThundeRiNG: Generating Multiple Independent Random \\ Number Sequences on FPGAs}


\renewcommand{\shortauthors}{H. Tan, et al.}

\begin{abstract}
In this paper, we propose ThundeRiNG, a resource-efficient and high-throughput system for generating multiple independent sequences of random numbers (MISRN) on FPGAs. Generating MISRN can be a time-consuming step in many applications such as numeric computation and approximate computing. Despite that decades of studies on generating a single sequence of random numbers on FPGAs have achieved very high throughput and high quality of randomness, existing MISRN approaches either suffer from heavy resource consumption or fail to achieve statistical independence among sequences.  In contrast, ThundeRiNG resolves the dependence by using a resource-efficient decorrelator among multiple sequences, guaranteeing a high statistical quality of randomness. Moreover, ThundeRiNG develops a novel state sharing among a massive number of pseudo-random number generator instances on FPGAs. The experimental results show that ThundeRiNG successfully passes the widely used statistical test, TestU01, only consumes a constant number of DSPs (less than 1\% of the FPGA resource capacity) for generating any number of sequences, and achieves a throughput of 655 billion random numbers per second. Compared to the state-of-the-art GPU library, ThundeRiNG demonstrates a $10.62\times$ speedup on MISRN and delivers  up to $9.15\times$ performance and $26.63\times$ power efficiency improvement on two applications ($\pi$ estimation and Monte Carlo option pricing). {This work is open-sourced on Github at~\url{https://github.com/Xtra-Computing/ThundeRiNG}.}
\end{abstract}





\maketitle

\newcolumntype{Y}{>{\centering\arraybackslash}X}
\newcolumntype{b}{X}
\newcolumntype{s}{>{\hsize=.5\hsize}X}

\DeclarePairedDelimiterXPP\seq[2]{}{\{}{\}}{_{#2}}{#1}
\DeclarePairedDelimiterX\set[1]\lbrace\rbrace{\def\given{\;\delimsize\vert\;}#1}


\section{Introduction}\label{sec:int}
A pseudo-random number generator (PRNG) generates a sequence of uniformly distributed random numbers. It is a fundamental routine at the core of many modern applications,
i.e., Monte Carlo simulation~\cite{rubinstein2016simulation,plimpton2019direct} and approximated graph mining~\cite{quoc2017approximate,222637,pavan2013counting}. Many of these applications are inherently parallel and can cope well with the increasing amount of data by mapping the parallelism onto modern hardware. This has led to the need to generate a massive quantity of pseudo-random numbers with high quality of statistical randomness. In other words, the PRNG itself must also be scalable~\cite{weigel2018monte}.

Field programmable gate arrays (FPGAs) have demonstrated promising performance on single sequence generation, benefiting from the good fit between the computation of PRNGs and the architecture of FPGAs.
PRNG generally adopts recurrence algorithms~\cite{everest2003recurrence} for generating a sequence of numbers and consists of two successive stages, as shown in the following equations.
\begin{align}
x_n &= f(x_{n-1}) \;  n=1,2,3,\ldots \label{eq:def1}\\
u_n &= g(x_n) \label{eq:def2}
\end{align}
The state $x_n$ belongs to $\boldsymbol{X}$,  which is a finite set of states (the state space). $f \colon \boldsymbol{X} \to \boldsymbol{X} $ is the state transition function.
$g \colon \boldsymbol{X} \to \boldsymbol{U} $ is the output function, where $U$ is the output space.
The generation of random numbers involves the following repeated steps: first, the state of the PRNG is updated according to~\Cref{eq:def1}; then, ~\Cref{eq:def2} extracts the random number $u_n$ from the state $x_n$.
To guarantee statistical randomness, {existing FPGA-based PRNGs~\cite{li2011software,dalal2008fast,thomas2012lut} usually implement the state transition with a large state space, requiring
{\em block RAMs} (BRAMs) in the FPGAs to be used as storage for state processing}.
The output stage usually includes bitwise operations such as truncation or permutation to increase the unpredictability of the sequence.
{Leveraging the bit-level customization capability of FPGAs, the algorithm-specific permutation can be efficiently implemented and pipelined with the state transition to achieve high throughput in FPGAs.}
For instance, previous studies~\cite{thomas2012lut,li2011software,dalal2008fast,bakiri2017ciprng,dabal2014study,dabal2012fpga} have shown that FPGAs deliver a better performance than CPU or GPU based single sequence generation.


\begin{table*}[t]
\centering
\caption{Survey of PRNG algorithms and implementations. The test suite for statistical quality is the TestU01 suite.}
\label{tab:cost_prng}
\resizebox{1\textwidth}{!}{%

\begin{threeparttable}
\begin{tabularx}{1.12\textwidth}{lccccccc}
\toprule
\multirow{2}{*}{\textbf{PRNG Algorithms}}
& \multirow{2}{*}{\textbf{Platform}}
& \multirow{2}{*}{\textbf{State width}}
& \multirow{1}{*}{\textbf{{\#Multiplication}}}
& \multirow{1}{*}{\textbf{Single sequence}}
& \multicolumn{2}{c}{\multirow{1}{*}{\textbf{Multiple sequences}}}

& \multirow{1}{*}{\textbf{Critical resources}}  \\

&&& \multirow{1}{*}{\textbf{{for $n$ instances}}}
&\multirow{1}{*}{\textbf{Statistical quality}}
&\multirow{1}{*}{\textbf{Methods}}
&\multirow{1}{*}{\textbf{Statistical quality}}
&\multirow{1}{*}{\textbf{on FPGAs}}\\ \midrule

Li et al.~\cite{li2011software}
&\multirow{3}{*}{FPGA}
&  $19937$ & 0 &   Crushable     & Substream   & Crushable    & Block RAMs\\

Dalal et al.\cite{dalal2008fast}
&
&  $19937$ & 0 &   Crushable     & Substream   & Crushable     & Block RAMs\\

LUT-SR~\cite{thomas2012lut}
&
&  $19937$ & 0 &   Crushable     & Substream   & Crushable    & LUTs     \\


\midrule

Philox4\_32~\cite{salmon2011parallel}
& \multirow{2}{*}{GPU/CPU}
&  $256$   & $ 6  n $ & Crush-resistant & Multistream &  Crush-resistant    & DSP slices\\

MRG32k3a~\cite{l1999good}
&
&  $384$   & $4  n$ & Crush-resistant & Substream   & Crushable    & DSP slices\\

\midrule

Xoroshiro128**~\cite{blackman2018scrambled}
&\multirow{3}{*}{CPU}
& $128$    & $ 2  n $ & Crush-resistant & Substream   &  Crush-resistant    & DSP slices\\

PCG\_XSH\_RS\_64~\cite{o2014pcg}
&
&  $64$    & $ n$ & Crush-resistant & Multistream & Crushable     & DSP slices\\

LCG64~\cite{marsaglia2003xorshift}
&
&  $64$    & $ n $ & Crushable       & Multistream  & Crushable    & DSP slices\\


\midrule

\textbf{ThundeRiNG}
& FPGA
& $192$    & $\bm{1}$ & \textbf{Crush-resistant}        & Multistream    & \textbf{Crush-resistant}    & LUTs\\
\bottomrule
\end{tabularx}%

\end{threeparttable}
}
\end{table*}


While the single sequence of random number generation on FPGAs has been well studied, extending it to generate multiple independent sequences of random numbers (MISRN) is nontrivial.
Despite that decades of studies on generating a single sequence of random numbers on FPGAs have achieved very high throughput and high quality of randomness, existing approaches for generating MISRN ~\cite{li2011software, dalal2008fast, durst1989using,hellekalek1998don, wu200625} either suffer from heavy resource consumption or fail to achieve independence among sequences.
First, the resources of FPGAs can easily become a limitation for the concurrent generation. The state transition stage usually adopts states with large space or complex nonlinear arithmetic operations in PRNG, which consumes the precious BRAM or DSP resources of FPGAs~\cite{li2011software, dalal2008fast}. Due to the heavy resource consumption, we cannot scale a large number of PRNG instances on a single FPGA.
Second, the correlation among multiple sequences leads to low quality of randomness~\cite{li2011software, dalal2008fast}. Sequences generated by the same type of PRNG tend to have correlation, diminishing the quality of randomness~\cite{durst1989using,hellekalek1998don}. In fact, the quality of the generated sequences of the previous two designs is not guaranteed~\cite{li2011software, dalal2008fast}, as they fail in some of the empirical tests such as TestU01~\cite{l2007testu01}.

To our best knowledge, none of the previous studies on FPGAs achieved high quality of randomness as required in many applications, or the high throughput and scalability for MISRN.
In this paper, we propose a high-throughput, high-quality, and scalable PRNG, called \textit{ThundeRiNG}, to tackle the aforementioned two challenges.
ThundeRiNG inherits {\em linear congruential generator}~\cite{lehmer1951mathematical} (LCG) that natively supports affine transformation to generate distinct sequences. {While the widely adopted LCG parallelization approaches such as state spacing~\cite{wu200625} suffer from long-range correlation~\cite{hellekalek1998don,durst1989using} and efficiency problems~\cite{deng2012large,walidainy2015improved}, we identified an opportunity to share the most resource-consuming stage between multiple PRNG instances on FPGAs, and found a technique to eliminate the correlation among the concurrently generated sequences.}

Specifically, ThundeRiNG makes the following contributions:
\begin{itemize}[leftmargin=*]
\item It enables state sharing for generating multiple independent sequences to solve the resource inefficiency problem when increasing the number of PRNGs instantiated on FPGAs.
\item It has a resource-efficient decorrelation mechanism to remove the correlation among sequences to guarantee the quality of randomness.
\item It consumes a constant number of DSPs for a varied number of generated sequences and achieves up to 655 billion
random numbers per second (20.95 Tb/s), without compromising the quality of randomness.

\item Compared with the state-of-the-art GPU implementation, it delivers up to $10.62\times$ performance improvement.
Furthermore, we demonstrate its effectiveness on two real world applications with delivering up to $9.15\times$ speedup on throughput and $26.63\times$ power efficiency.
\end{itemize}

The rest of the paper is organized as follows. Section~\ref{sec:background} introduces the background and related work. Section~\ref{sec:proposed_method} presents the design, followed by the implementation details on FPGA in Section~\ref{sec:imp}. We present the experimental results and case studies in Sections~\ref{sec:eva} and~\ref{sec:case}, respectively. We conclude this paper in Section~\ref{sec:conclusion}.


\section{Background and Related Work}\label{sec:background}

In this section, we present the quality criteria and review existing approaches for generating MISRN (summarized in~\Cref{tab:cost_prng}).

\subsection{PRNG Quality Criteria} \label{subsec:randomness}
The statistical randomness of the generated sequences is the most important quality criterion of PRNG.

\noindent
\textbf{Statistical Randomness.}
Randomness is hard to measure due to its considerable evaluation space. Instead, statistical randomness is commonly used for measuring the quality of a PRNG. A numerical sequence is statistically random if it contains no recognizable pattern or regularities~\cite{wiki:sr}. In essence, statistical randomness indicates how well the successive outputs of PRNG behave as independent and identically distributed (i.i.d) random variables. 

\noindent
\textbf{Statistical Randomness Testing.}
There are two testing approaches for statistical randomness: theoretical test and empirical test. Theoretical test is a kind of prior test based on the knowledge of the PRNG algorithm, and thus it is not applicable for PRNGs without clear mathematical modeling~\cite{10.5555/270146}. In contrast, the empirical test is able to extract recognizable patterns from the generated sequences without knowledge of detailed mathematical modeling, and it is widely adopted in the evaluation of PRNGs~\cite{marsaglia1995diehard,bassham2010sp,brown1994security}.

The TestU01 suite\cite{l2007testu01}, which is the most stringent empirical test suite, has been widely used and has become the standard for testing the statistical quality of a PRNG. It contains several test batteries, including SmallCrush (with 10 tests), Crush (with 96 tests), and BigCrush (with 160 tests). PRNGs that pass all tests in those test batteries can be referred as \textbf{\textit{crush-resistant}}, indicating a good quality of statistical randomness, while the PRNGs fail to do that is called \textbf{\textit{crushable}}, indicating that recognizable patterns exist~\cite{salmon2011parallel}. All FPGA-based PRNGs (except this work) in~\Cref{tab:cost_prng} are crushable even for single sequence generation.

\subsection{Multiple Sequence Generation Methods}\label{subsec:lcg}
In supporting MISRN, existing PRNGs usually adopt one of the two methods: substream and multistream (as shown in the column \textbf{Methods} for ``Multiple sequences`` in Table~\ref{tab:cost_prng}).

\noindent
\textbf{Substream.}
{Substream based solutions equally divide the state space into many non-overlapped subspaces to generate disjoint logical sequences.
The practical criterion to guarantee nonoverlapping is maintaining at least $2^{63}$ skipped elements among the logical sequences}~\cite{blackman2018scrambled}.
This method is widely adopted in existing works~\cite{matsumoto1998mersenne,haramoto2008efficient,panneton2006improved,blackman2018scrambled,li2011software}.

\noindent
\textbf{Multistream.}
The multistream approach is that the same PRNG module is instantiated multiple times, and the instances run concurrently with different parameters for generating multiple distinct streams.
All prior cited FPGA-based PRNGs use the substream solution, and only CPU/GPU based solutions adopt multistream based solutions,
e.g., Philox4\_32~\cite{salmon2011parallel} and PCG\_XSH\_RS\_64~\cite{o2014pcg}.

\subsection{Challenges of MISRN Generation on FPGAs}

\Cref{tab:cost_prng} summarizes the existing FPGA-based PRNGs as well as CPU/GPU based PRNGs. We revisit those algorithms for potential adoption and thus analyze each method in terms of the state width, number of multiplications, statistical quality, multi-sequence generation method, and critical resources. We identify the limitations of existing works and the open challenges of multi-sequence generation on FPGAs.

\subsubsection{Challenge 1: correlation among sequences.}
{A common issue with existing methods of MISRN is the correlation between sequences. This leads to poor statistical randomness and may not satisfy the application requirements~\cite{entacher1999parallel,hellekalek1998don}.
Correlation violates the independence of the generated sequences and leads to inaccurate or biased results even in the simplest applications~\cite{entacher1999parallel,hellekalek1998don}.} All FPGA-based solutions are crushable for both single and multiple sequence generation.

\subsubsection{Challenge 2: high throughput via parallelism.}
To increase throughput, multiple pseudo-random sequences must be generated concurrently. The recent methods~\cite{blackman2018scrambled,matsumoto1998mersenne,li2011software,o2014pcg,salmon2011parallel} require instantiating one PRNG module for generating one sequence. On FPGAs, this translates to a significant resource consumption that is linearly proportional to the number of sequences. As a single FPGA has limited resources, this will severely limit the number of sequences that can be generated concurrently on one FPGA. Even worse, PRNGs usually require either a large state width (e.g., the 19937-bit state of FPGA-based solutions shown in~\Cref{tab:cost_prng}), or complex arithmetic (e.g., the multiplication operation of CPU-based solutions shown in~\Cref{tab:cost_prng}) to improve the randomness of the output.
When instantiating on FPGAs, the large state width will consume the BRAM resources of FPGAs, and the complex arithmetic consumes heavily on DSPs. As a result, directly implementing existing CPU/GPU-based PRNGs on the FPGAs can be resource- and throughput-constrained.




\newcommand{\indep}{\perp \!\!\! \perp}
\newcommand{\notindep}{\not\!\perp\!\!\!\perp}
\newcommand{\Mod}[1]{\ (\mathrm{mod}\ #1)}

\newenvironment{conditions*}
  {\par\vspace{\abovedisplayskip}\noindent
   \tabularx{\columnwidth}{>{$}l<{$} @{}>{${}}c<{{~}$}@{} >{\raggedright\arraybackslash}X}}
  {\endtabularx\par\vspace{\belowdisplayskip}}

\section{Design of ThundeRiNG}\label{sec:proposed_method}

We describe the design of our proposed ThundeRiNG in this section, followed by the implementation details on FPGA in the next section. As far as we know, ThundeRiNG is the first FPGA-based solution that solves the two above-mentioned challenges, providing a high-throughput, high-quality PRNG.
ThundeRiNG is based on a well-studied {\em linear congruential generator} (LCG) PRNG~\cite{lehmer1951mathematical}.
To ensure the highest quality of the output, ThundeRiNG adopts a resource-efficient decorrelator that removes the correlation between the multiple sequences generated by parameterizing the LCG via increments.
To scale the throughput, ThundeRiNG uses state sharing to reduce resource consumption when instantiating a massive number of PRNG instances on FPGAs.

\subsection{Parameterizing LCG via Increment } \label{subsec:lcgincrement}
The LCG algorithm has three parameters, labelled as $m$, $a$ and $c$, where $m$ is the \textbf{\textit{modulus}} ($m \in \mathbb{Z}^{+} $), $a$ is the \textbf{\textit{multiplier}} ($0 < a < m$) and $c$ is the \textbf{\textit{increment}} ($0 \leq c < m$). The set of sequences generated with the same $a, m, c$ parameters are represented as $\mathbb{X}_{a,m} = \lbrace X_{a,m}^{c}, \: ...\mid{}  {0 \leqslant c < m}  \rbrace$, and the generation of each instance of $X_{a,m}^c$ is defined in the following equation:
\begin{align}
 x_{n+1} &= (a \cdot x_{n} + c)\mod{ m }, ~~n \geq 0  \label{eq:lcg_s} \\
 u_{n+1} &= truncation(x_{n+1}) \label{eq:lcg_o}
\end{align}

\Cref{eq:lcg_s} is the state transition function, and
~\Cref{eq:lcg_o} is the output function, which conducts a simple truncation on $x_{n}$ to guarantee that the state space is larger than the output space~\cite{o2014pcg}.

Instead of parameterizing PRNGs with modulus and multiplier~\cite{durst1989using}, ThundeRiNG explores parameterization via the increment $c$, to enable a resource-efficient multiple sequence generation method. However, the generated distinct sequences also suffer the severe correlation problem~\cite{percus1989random}, which motivates us to develop a decorrelation approach in the following subsection.


\subsection{Decorrelation}\label{subsec:decorrelation}

As shown in previous studies~\cite{percus1989random}, the sequences generated by LCGs with different increments still have severe correlations.
There have been several existing approaches to eliminate the correlations, such as dynamic principal component analysis~\cite{rato2013advantage} or Cholesky matrix decomposition~\cite{kessy2018optimal,ghanem2003stochastic}.
However, these methods involve massive computation, which is resource inefficient and unpractical for high-throughput random number generation on FPGAs.

{The number of possible combinations of sequences generated by LCG is very large, leading the correlations among multiple sequences hard to analyze}. Therefore, we first consider the correlation between two sequences to simplify the problem, and then we extend it to multiple sequences.

\subsubsection{Decorrelation on two sequences}
Yao's XOR Lemma~\cite{yao1982theory} states that the hardness of predication is amplified when the results of several independent instances are coupled by the exclusive disjunction.
Therefore, we use the XOR operation to amplify the independence of the generated sequences by LCG.
Specifically, we first adopt a light-weight but completely different algorithm from LCG for the generation of two sequences, even if they are weakly correlated, and then combine them to the sequences generated by the LCG algorithm with bitwise XOR operations.

\Cref{thm:1} gives the theoretical proof of the improved independence for the newly generated sequences with our approach.


\begin{theorem} \label{thm:1}

{Suppose $X_{a,m}^{c_1} = \seq{x_{n}^{c_1}}{n\in\mathbb N}$  and $X_{a,m}^{c_2} = \seq{x_{n}^{c_2}}{n\in\mathbb N}$ are two distinct sequences belong to $\mathbb{X}_{a,m}$, 
and there are two weakly correlated sequences $I = \seq{i_n}{n\in\mathbb N}$ and $J = \seq{j_n}{n\in\mathbb N}$ 
,which are uncorrelated with the sequences in $\mathbb{X}_{a,m}$.
Then the correlation between the combined sequences, $ Z^{1} = \seq{x_{n}^{c_1} \oplus i_n}{n\in\mathbb N} $ and $ Z^{2}= \seq{x_{n}^{c_2} \oplus j_n}{n\in\mathbb N} $ is weaker than the correlation between $X_{a,m}^{c_1}$ and $X_{a,m}^{c_2}$.}

\end{theorem}

\begin{proof}[Proof]
    First, we consider two binary uniformly distributed sequences, $X$ and $Y$. As we cannot directly calculate the probability based on XOR, we transform the XOR operator to multiplication~\cite{goldreich1995yao}. Specifically, {we define a sequence transformation, $h(X) = 1 - 2X = \seq{1 - 2 \cdot x_{n}}{n \in \mathbb N} $, which maps the value of the elements in $X$ from $\set{0,1}$ to $\set{-1, 1}$ }. Then we have
    \begin{align}
        h(X \oplus Y ) =h(X) \cdot h(Y).
    \end{align}
    The mathematical expectation ($E$), variance ($var$) of $h(X)$, and the covariance ($cov$) between $h(X)$ and $h(Y)$ are calculated as follows:
    \begin{align}
        E(h(X)) &= 1 - 2E(X) \label{eq:trick_1}\\
        var(h(X)) &= 4 var(X) \label{eq:trick_2} \\
        cov(h(X),h(Y)) &= 4 cov(X, Y) \label{eq:trick_3}
    \end{align}
    As X and Y are uniformly distributed, then we have
    \begin{align}
        E(X) = \mu_{X} \approx 1/2 \label{eq:e_approx}.
    \end{align}
    Since $var(X) = E(X^2) - (E(X))^2$, we can approximate the variance of X:
    \begin{align}
        var(X) = \mu_{X} - (\mu_{X})^2 \approx 1/4 \label{eq:approx_var}
    \end{align}
    Therefore, we can calculate the variance of the new sequence $X \oplus Y$ by ~\Cref{eq:trick_2,eq:approx_var}:
    \begin{align}
        var(X \oplus Y) &= \frac{var(h(X \oplus Y))}{4} \nonumber
        = \frac{1}{4} \cdot \bigg(var(h(X)) \cdot var(h(Y))   \nonumber\\
        & + var(h(Y))\cdot E[h(X)]^2 +var(h(X))\cdot  E[h(Y)]^2\bigg)\nonumber\\
        & = \frac{var(h(X)) \cdot var(h(Y))}{4}   \nonumber\\
        &= 4\cdot var(X) \cdot var(Y)  \approx 1/4
        \label{eq:approx_xor_var}
    \end{align}
    Taking two sequences in $Z$(in the Theorem definition), their correlation $\rho_{Z}$ can be represented by the definition of correlation:
    \begin{align}
        \rho_{Z} =\frac{cov(X_{a,m}^{c_1} \oplus I, X_{a,m}^{c_2} \oplus J)}{\sqrt{var(X_{a,m}^{c_1} \oplus I) \cdot var(X_{a,m}^{c_2} \oplus J) }}
        \label{eq:rho}
    \end{align}
    \Cref{eq:rho} can be further approximated using~\Cref{eq:approx_xor_var}:
    \begin{align}
        \rho_{Z} \approx 4 \cdot cov(X_{a,m}^{c_1} \oplus I, X_{a,m}^{c_2} \oplus J)  \label{eq:approx_var_res}
    \end{align}
    where the covariance can be rewritten as
    \begin{align}
         cov(&X_{a,m}^{c_1} \oplus I, X_{a,m}^{c_2} \oplus J) \nonumber\\
        &= \frac{cov[h(X_{a,m}^{c_1} \oplus I), h(X_{a,m}^{c_2} \oplus J)]}{4} \nonumber\\
        &= \frac{cov[h(X_{a,m}^{c_1}) \cdot h(I), h(X_{a,m}^{c_2}) \cdot h(J)]}{4}  \nonumber\\
        &=
        \!\begin{aligned}[t]
            &\frac{1}{4} \cdot \bigg(E[h(X_{a,m}^{c_1}) h(X_{a,m}^{c_2}) h(I) h(J)] \nonumber\\
            &- E[h(X_{a,m}^{c_1}) h(I))] E[h(X_{a,m}^{c_2}) h(J)] \bigg).
        \end{aligned}
        &\\ \label{eq:xor_simple_1}
    \end{align}
    As $X_{a,m}$ is independent of $I$ and $J$, the first item in~\Cref{eq:xor_simple_1} can be represented by their covariances:
    \begin{align}
        E[&h(X_{a,m}^{c_1}) h(X_{a,m}^{c_2}) h(I) h(J)]  \nonumber\\
        &=E[h(X_{a,m}^{c_1}) h(X_{a,m}^{c_2})] \cdot E[h(I) h(J)] \nonumber\\
        &=
        \!\begin{aligned}[t]
            &\bigg(cov[h(X_{a,m}^{c_1}),h(X_{a,m}^{c_2})] + E[h(X_{a,m}^{c_1})]E[h(X_{a,m}^{c_2})] \bigg) \nonumber\\
            &\cdot  \bigg(cov[h(I),h(J)] + E[h(I)]E[h(J)] \bigg)
        \end{aligned}
        &\\ \label{eq:cov_intermiddle}
    \end{align}
    Similar with~\Cref{eq:approx_var_res}, we use the correlation ($corr$) to replace the covariance items in~\Cref{eq:cov_intermiddle}:
    \begin{align}
        E[&h(X_{a,m}^{c_1}) h(X_{a,m}^{c_2}) h(I) h(J)]  \nonumber\\
        &\approx
        \!\begin{aligned}[t]
            &\frac{1}{4} \cdot \bigg(corr[h(X_{a,m}^{c_1}),h(X_{a,m}^{c_2})] + 4E[h(X_{a,m}^{c_1})]E[h(X_{a,m}^{c_2})] \bigg) \\
            &\cdot  \bigg(corr[h(I),h(J)] + 4E[h(I)]E[h(J)] \bigg)
        \end{aligned}
        \label{eq:corr_intermiddle}
    \end{align}
    As $X_{a,m}^{c_1}$ and $X_{a,m}^{c_2}$ from the same LCG set, the difference of their expectations can be ignored. We use $\mu_{X_{a,m}}$ to represent their expectation:
    \begin{align}
        E(X_{a,m}^{c_1}) = E(X_{a,m}^{c_2}) =\mu_{\mathbb{X}_{a,m}}
    \end{align}
    Here, $\rho_{\mathbb{X}_{a,m}}$ represents the correlation of $X_{a,m}^{c_1}$ and $X_{a,m}^{c_2}$, and $\rho_{(I,J)}$ represents the correlation of $I$ and $J$, Finally, combining ~\Cref{eq:approx_var_res,eq:xor_simple_1,eq:cov_intermiddle,eq:corr_intermiddle},~\Cref{eq:rho} can be simplified as
    \begin{align}
       \rho_{Z} &\approx \rho_{X_{a,m}} \cdot \rho_{(I,J)} + \rho_{X_{a,m}} \cdot (1 - 2\mu_{(I,J)}) + \rho_{(I,J)} \cdot (1 - 2\mu_{\mathbb{X}_{a,m}}).
    \label{eq:final}
    \end{align}
    The expectations of the sequence $I$, $J$ and $X_{a,m}$ are very close to $1/2$. Hence, the terms $\rho_{\mathbb{X}_{a,m}} \cdot (1 - 2\mu_{(I,J)})$ and  $\rho_{(I,J)} \cdot (1 - 2\mu_{\mathbb{X}_{a,m}})$ in~\Cref{eq:final} can be ignored. Finally, the correlation between $Z^{1}$ and $Z^{2}$ is simplified as
    \begin{align}
        \left| \rho_{Z} \right| & \approx \left| \rho_{X_{a,m}}\right| \cdot \left| \rho_{(I,J)} \right|. \label{eq:end}
    \end{align}

We assume that the sequences in the $\mathbb{X}_{a,m}$ are highly correlated,  so $\left| \rho_{\mathbb{X}_{a,m}}\right|$ is close to 1. $\left| \rho_{(I,J)}\right|$ is a small value ($< 1$) as the sequences $I$ and $J$ are only weakly correlated. Thus,  $\left|\rho_{Z}\right| <  \left| \rho_{\mathbb{X}_{a,m}}\right|$ from~\Cref{eq:end}, which proves that our XOR based approach decreases the correlation of the original LCG generated sequences.
\end{proof}

\subsubsection{Decorrelation on multiple sequences}
On top of~\Cref{thm:1}, we further consider the correlation among the sequences generated from more than two generators.
Typically, the independence of multiple random number sequences has two measurements, mutual independence and pairwise independence.
Mutual independence is a strong notion of independence.
It requires that each sequence is independent of all other sequences and any combination of other sequences in the set. Matsumoto et al.~\cite{matsumoto2000dynamic} have the hypothesis of mutual independence on the linear recurrence. However, there is no mathematically rigorous proof, and it is even impossible to evaluate all possible combinations in empirical tests when the number of sequences is large.
Pairwise independence indicates that any two sequences in the domain are independent of each other, which is mostly considered measurement in PRNG~\cite{l2007testu01}. Therefore, we only extend~\Cref{thm:1} to pairwise independence of multiple sequences, as given in~\Cref{thm:2}.


\begin{theorem} \label{thm:2}

Suppose there is a set of independent sequences, denoted as $\mathbb{K}$, of which all members are uncorrelated with $\mathbb{X}_{a,m}$.
Separately selecting $n$ sequences from $\mathbb{K}$ and $\mathbb{X}_{a,m}$, denoted as $\{K^i\}_{i=1}^{n}$ $(n>2)$ and $\{X_{a,m}^{c_i}\}_{i=1}^{n}$. The combined sequences, $\{Z^{i}= X_{a,m}^{c_i} \oplus K^i\}_{i=1}^{n}$ are pairwise independent from each other.

\end{theorem}

\begin{proof}[Proof]

We give the induction proof. Let $P(n)$ be the statement "the $n$ decorrelated sequences $\{Z^{i}= X_{a,m}^{c_i} \oplus K^i\}_{i=1}^{n}$ are pairwise independent". We will prove that $P(n)$ is true for all $n>2$. 

We first prove $P(3)$ is true, which means three distinct sequences are pairwise independent (denote them as $Z^1$, $Z^2$ and $Z^3$).
Based on our definition, we have $Z^1 =X_{a,m}^{c_1} \oplus K^1$ and $Z^2 =X_{a,m}^{c_2} \oplus K^2$. $K^1$ and $K^2$ belong to $\mathbb{K}$.
Hence based on~\Cref{thm:1}, we can get that $Z^1$ is independent from $Z^2$. In a similar way, $Z^1$ is independent from $Z^3$, and $Z^2$ is independent from $Z^3$. Therefore, $P(3)$ is true.

For the inductive step, assume $P(j)$ true, we will prove that $P(j + 1)$ is also true. Compared to $P(j)$, $Z^{j + 1}$ is the newly introduced sequence, of which the $K^{j+1}$ is independent with the sequences in $\{K^i\}_{i=1}^{j}$. 

Successively combining $Z^{j + 1}$ with all sequences in $\{Z^{i}\}_{i=1}^{j}$ into pairs, and adopting~\Cref{thm:1} on all of the pairs, we derive that $Z^{j + 1}$ is independent from all the members in  $\{Z^{i}\}_{i=1}^{j}$. Therefore, according to the definition of pairwise independence,  $P(j+1)$ is true, completing the proof.
\end{proof}

\begin{figure}[t]
     \centering
     \includegraphics[width=1\columnwidth]{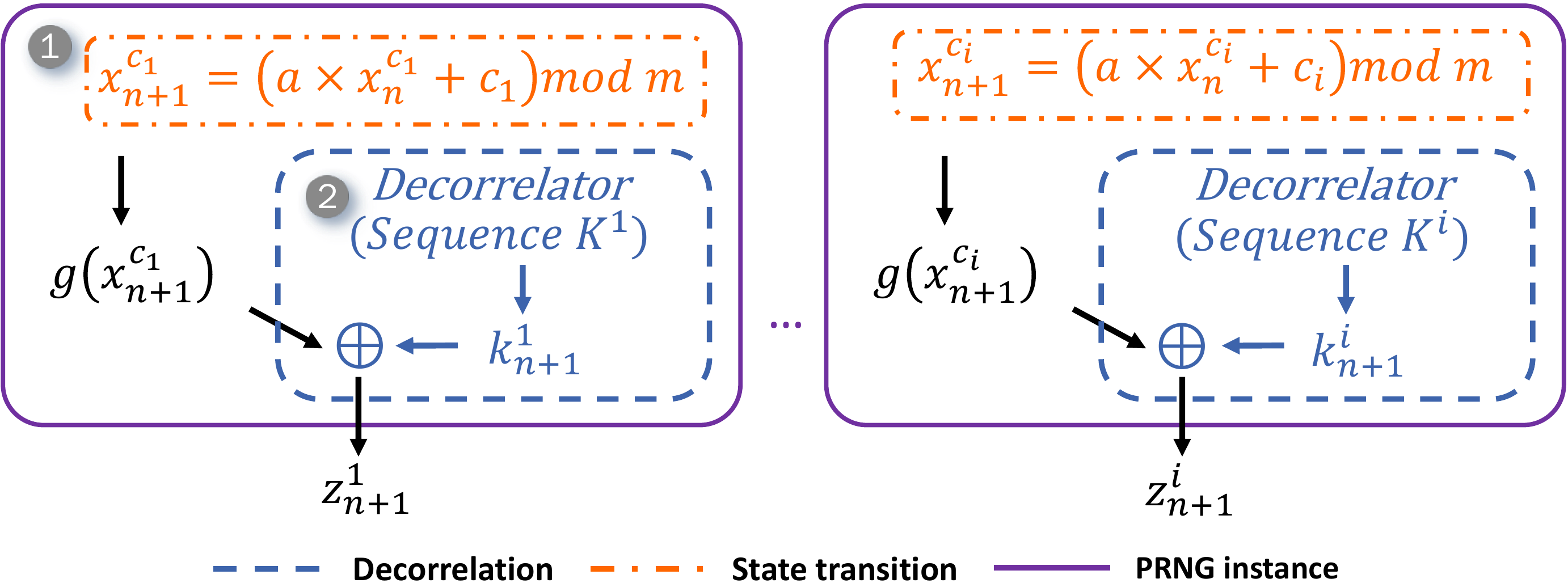}
     \caption{The decorrelation mechanism.}
     \label{fig:decorrelation}

\end{figure}

\subsubsection{The algorithm}
We propose our decorrelation method based on~\Cref{thm:1} and~\Cref{thm:2}, using the XOR operation to combine the correlated LCG sequences with a series of independent sequences.

The algorithm flow is shown in~\Cref{fig:decorrelation},
for the $\bm{i}$-th instance at Step $\bm{n+1}$ , it transits the old state from $x_{n}^{c_i}$ to the new state $x_{n+1}^{c_i}$ by the LCG algorithm, and the inside decorrelator generates an element $k_{n+1}^{i}$ and an XOR operation is performed to output $z_{n+1}^{i}$.

According to the proof in~\Cref{thm:2}, there are two theoretical constraints on the sequences generated by the decorrelator. Define $\mathbb{K}$ to be the set of all candidate sequences for the decorrelator. The constraints are:
\begin{enumerate*}[label=(\roman*)]
  \item Every sequence in~$\mathbb{K}$ should be independent of the sequences in LCG set $\mathbb{X}_{a,m}$.
  \item Any pair of sequences in~$\mathbb{K}$ should not be strongly correlated with each other.
\end{enumerate*}
These are the fundamental guidelines for choosing a suitable decorrelator.

Beyond the theoretical constraints, several practical factors could also be considered to reduce the selection space of the decorrelator when implementing the decorrelation method on FPGAs.
First, since the XOR operation could reduce the statistical bias~\cite{yao1982theory}, the sequences generated by the decorrelator do not need to have perfect statistical randomness.
Second, the decorrelator should be lightweight to be resource-efficient. For example, it is desirable to reduce the number of multiplications that are costly in FPGA resource consumption.
Finally, the decorrelator should produce massive uncorrelated sequences, adapting to the different required degree of parallelism.

Based on the theoretical constraints and practical considerations, we adopt the xorshift algorithm~\cite{marsaglia2003xorshift} as the decorrelator for the following reasons.
First, the generation process of xorshift is completely independent of the LCG algorithm, which guarantees the first theoretical constraint.
Second, xorshift supports the substream method to generate long-period logical sequences, which can avoid long-range correlations~\cite{haramoto2008efficient} and hence meets the second theoretical constraint.
Lastly, as xorshift is based on the binary linear recurrence, it only uses bit-shift and XOR operations that can be efficiently implemented on FPGAs.



\subsection{State Sharing Mechanism}\label{subsec:statesharing}

{The decorrelation method introduced in the previous section solves the correlation issue of the LCG seqences $\mathbb{X}_{a,m}$. 
To use it on FPGA platforms, each LCG sequence is generated by an independent calculation process, and it requires one multiplication and one addition operation during each step of the state transition. This can require a large number of hardware multipliers in MISRN.}
{To reduce the number of hardware multipliers, we propose the state sharing mechanism that reuses the intermediate results of state transition over distinct generators.} 

In order to enable intermediate result reuse, we further extend the LCG transition. Considering that $\xi_h$ is an addition transition with a given constant integer $h$ after the LCG transition (given in ~\Cref{eq:lcg_s}): 
\begin{align}
w_{n } =\xi_h(x_{n}) = (x_{n} + h)\mod{m} \label{eq:lt}
\end{align}
We can get the transition from $w_{n+1}$ to $w_{n}$ by expanding the modulus and replacing $x_{n}$ with~\Cref{eq:lcg_s}. That is
\begin{align}
w_{n+1} = [a \cdot w_n + (l\cdot m + c - a \cdot h)] \mod{m}, 
\end{align}
where $l$ is an integer introduced during the expansion of the modulus operation. The transition from $w_{n+1}$ to $w_{n}$ has the same form as LCG. As the multiplier $a$ is the same as the multiplier in~\Cref{eq:lcg_s}, the sequence $W$ generated by~\Cref{eq:lt} belongs to $\mathbb{X}_{a,m}$. This indicates that selecting different $h$ results in a unique sequence, which is the same as changing the increment $c$ of LCG. Hence, we have the following representation:
\begin{align}
w_{n+1} = \xi_h( \varphi_c(x_n)) =  \varphi_{(l\cdot m + c - a \cdot h)}(w_n) \label{eq:rt}
\end{align}
where $\varphi_c$ is the LCG transition with a specific increment $c$.

\Cref{eq:rt} allows us to share the output of $\varphi_c$ in the generation of multiple sequences. We illustrate this state sharing mechanism, as shown in~\Cref{fig:state_sharing}. We refer $\varphi_c$ as the root transition that can be shared among multiple sequences, and $\xi_h$ as the leaf transition. Every leaf transition uses the same root state from a root transition and occupies a unique number $h$ to guarantee that the sequences generated are distinct from others.

Combining with our decorrelation method, the flow of state sharing MISRN generation is as follows. First, the root transition generates an intermediate state at step $n+1$, $x_{n + 1}$.
Then, it is shared among $p$ instances.
Furthermore, the $i$-th instance transits to $x_{n + 1}$ by a unique leaf transition $\xi_{h_{i}}$ to get a unique leaf state $w_{n+1}^{i}$.
Finally, $w_{n+1}^{i}$ goes through the output stage and then couples with the corresponding output of the decorrelator (illustrated in~\Cref{subsec:decorrelation}) to produce a random number ${z'}_{n+1}^{i}$.



\begin{figure}[h]
     \centering
     \includegraphics[width=1\columnwidth]{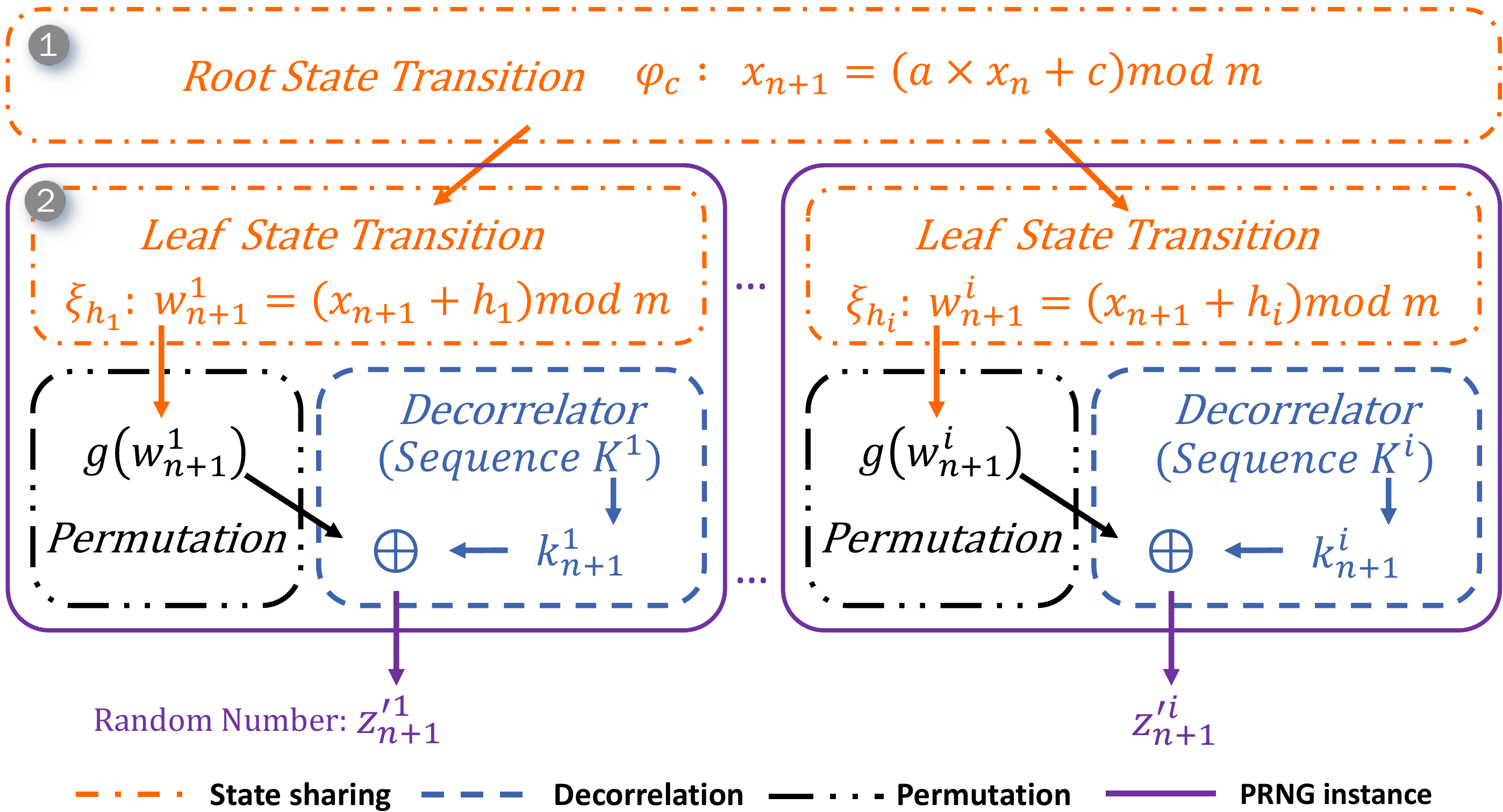}
     \caption{The State sharing mechanism.}
     \label{fig:state_sharing}
     \vspace{-10pt}
\end{figure}

To guarantee a maximum period of the sequences generated from the leaf transition, there is a constraint on the selection of $h$.
First, referring to Hull-Dobell Theorem~\cite{hull1962random}, $(l \cdot m + c - a \cdot h)$ 
must be an odd number. 
Second, as $m$ is power-of-two, $l \cdot m$ is always even. Hence, we only need to consider the parity of~$(c - a \cdot h)$.
Again, $c$, the increment in the root transition, which is also under the constraint of the Hull-Dobell Theorem, is an odd number.
Therefore, if $a \cdot h$ is an even number, $(l \cdot m - a \cdot h)$ is an odd number, and the Hull-Dobell Theorem holds. Finally, because $a$ is a prime number, an even $h$ will let the Hull-Dobell Theorem hold to guarantee the maximum period.


Comparing with all existing methods, which need $p$ times multiplication for $p$ distinct instances to transit the state at each step, our state sharing method \textbf{needs only one multiplication} along with $p$ addition operations.
Specifically, on the FPGA platform, our approach only needs one multiplier to support any number of PRNG instances.
This completely resolves the bottleneck of existing methods to increase the number of high-quality PRNG instances to improve the throughput.



\subsection{Permutation Function for Output} \label{subsec:permutation}

Since LCG is known to have weak statistical quality of low-order bits~\cite{l1999tables},
ThundeRiNG adopts the random rotation permutation in the output function $g$ as proposed by O’Neill~\cite{o2014pcg}. 
Basically, it performs a uniform scrambling operation and then remaps the bits to enhance the statistical quality. 
{The remapping operation is a bitwise rotation, of which the number of rotations is determined by the leaf state. 
As the leaf states are different from each other, the rotation operations of different sequences are different, which can further reduce the collinearity.}
We demonstrate the impact of the adopted permutation in~\Cref{subsubsec:correlation_evalutation}.

\newcommand*\circled[1]{%
}

\newcommand*\circledd[1]{\tikz[baseline=(char.base)]{
    \node[shape=circle,draw,inner sep=1pt] (char) {#1};}}

\section{Implementation of ThundeRiNG}\label{sec:imp}



\begin{figure}[]
     \centering
     \includegraphics[width=1\columnwidth]{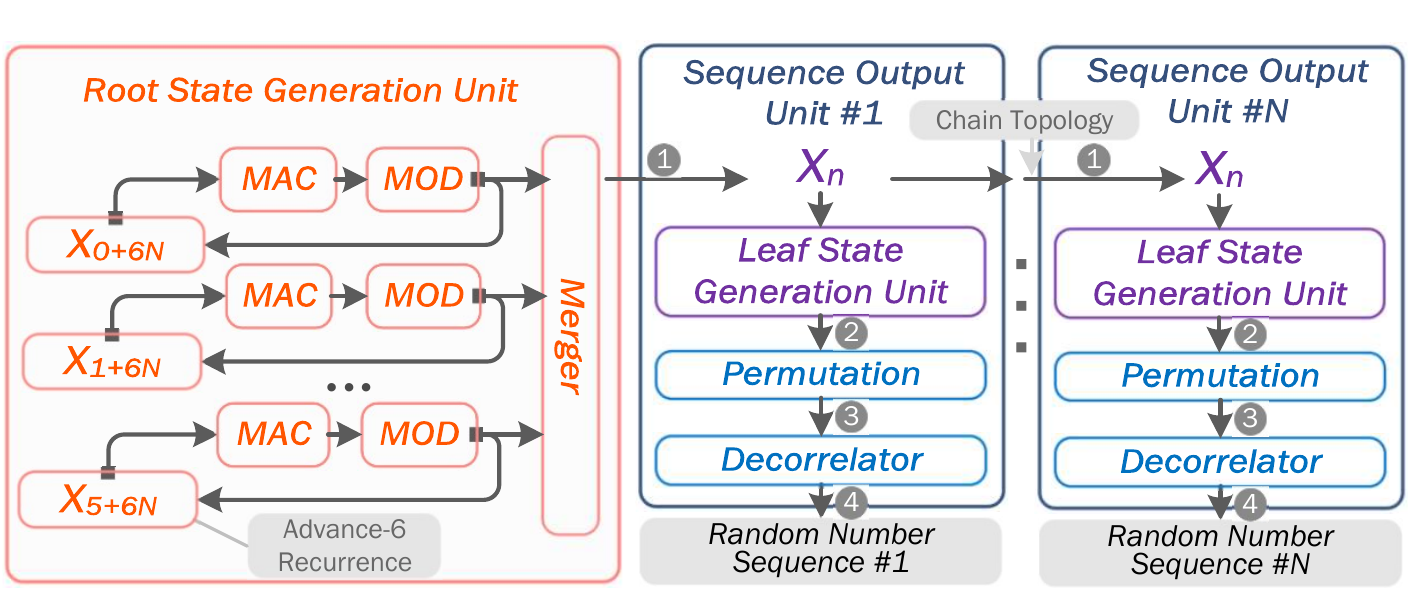}
     \caption{The overview of implementation of ThundeRiNG.}
     \label{fig:architecture_overview}
     
\end{figure}

\subsection{Architecture Overview}

The overall architecture of ThundeRiNG implementation on FPGA is shown in \Cref{fig:architecture_overview}. 
The architecture mainly consists of one {\em root state generation unit} (\textbf{RSGU}) and multiple {\em sequence output unit} (\textbf{SOU}). Each SOU is responsible for generating a single random sequence.
RSGU generates one root state per cycle by performing the transition function of LCG and shares it with all SOUs through a daisy-chaining topology~\cite{boutros2018you} interconnection (Step \ding{172}).
SOU is the main component to realize the decorrelation and state sharing as presented in Section~\ref{sec:proposed_method}. It is composed of three subcomponents, including the {\em leaf state generation unit} (LSGU), permutation unit, and decorrelator unit.
Given the shared root state $x_n$, each SOU executes the leaf state generation in LSGU to generate a leaf state (Step \ding{173}) and then executes the permutation function to generate a distinct LCG sequence (Step \ding{174}).
Finally, the decorrelator will output the random number sequence with the input from the permutation (Step \ding{175}).

\subsection{Root State Generation Unit (RSGU)} \label{subsec:stu}
The RSGU recursively generates root states by the following equation: $x_{n+1} =( a \cdot x_{n} + c) \mod{m}$, where the initial $x_{n}$ is initialized with a random constant.
Although the computation of one multiply-accumulate operation (MAC) is rather simple, generating one root state per cycle is nontrivial due to the true dependency~\cite{10.5555/502981} introduced by the recursive computation pattern.
The calculation of the next state is based on the previous state, which means the calculation of a new state can only be started after the calculation of the previous state is finished.
However, the latency of multiplication with DSPs on FPGA usually has multiple cycles (e.g., DSP48E2 takes 6 cycles as indicated in ~\Cref{fig:stu}(a).)
As a result, the throughput is limited by the long latency of multiplication on FPGAs.

Another possible way of finishing the MAC in one cycle is to use logic resources such as LUTs and registers to construct the MAC directly.
However, a large-bit multiplier costs a large number of LUTs.
Moreover, it consists of many long combinational logic paths
that result in a large propagation delay and lead to low frequency.
As shown in~\Cref{fig:stu}(b), the LUT-based 64-bit MAC can output $x_1$ after $x_0$ in one cycle. However, it runs at a much lower frequency, which degrades the generation performance.

{
Instead, ThundeRiNG hides the long multiplication latency of DSPs by leveraging the step-jump-ahead feature of LCG~\cite{brown1994random}.
Although recursive state generation has flow dependency, LCG supports the arbitrary advance recurrence, which could generate states in a jump-wise manner: $x_{n+i} = f_{adv-i}(x_{n})$, where $i$ is an arbitrary integer number and $f_{adv-i}$ is a new recurrence function whose parameters can be derived from the original LCG state transition function (\Cref{eq:lcg_s}).}
With it, the process of state generation can be partitioned into multiple portions and executed by multiple hardware units with different lags in parallel.
For example, to generate $\{x_0, x_1,...x_{6},x_{7}...\}$ as the output state sequence,
instead of using one state generator, we could use two dependent advance-2 state generators with the first one generating the state sequence $\{x_0, x_2, ..., x_{6},...\}$ and the second one generating the state sequence $\{x_1, x_3, ..., x_{7}, ...\}$.

The architecture of RSGU is shown in~\Cref{fig:architecture_overview}.
{The RSGU consists of multiple independent state generators. Each state generator contains a MAC unit for multiplication operations, a modulus unit for modulus operation, and registers for storing the temporal state during the recurrence process. 
State generators perform the same $f_{adv-i}$ on adjacent starting positions in parallel. Their outputs are further merged in the order of the original state sequence. 
Since the latency of state calculation is six cycles, we implement six state generators so that RSGU could generate one state per cycle, as indicated in~\Cref{fig:stu}(c).}

In addition, our design does not involve critical combinational logic paths, and the evaluation shows that the post-routing frequency can be scaled up to 550MHz using the HLS toolchain.
Parameters for advance-$i$ recurrence are calculated in compile-time following the algorithm proposed by Brown et al.~\cite{brown1994random}, which has $O(\log(i))$ complexity.
With $i$ equal to 6 in our case, the overhead of calculating these parameters is negligible.





\begin{figure}[t!]
    \centering
    \begin{subfigure}[t]{0.45\textwidth}
        \centering
        \includegraphics[width=1\linewidth]{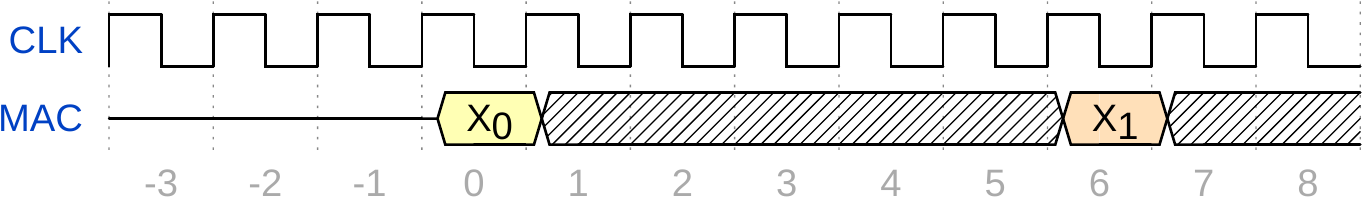}
        \caption{DSP-based state generation.}
    \end{subfigure}
    \smallskip
     \begin{subfigure}[t]{0.45\textwidth}
        \centering
        \includegraphics[width=1\linewidth]{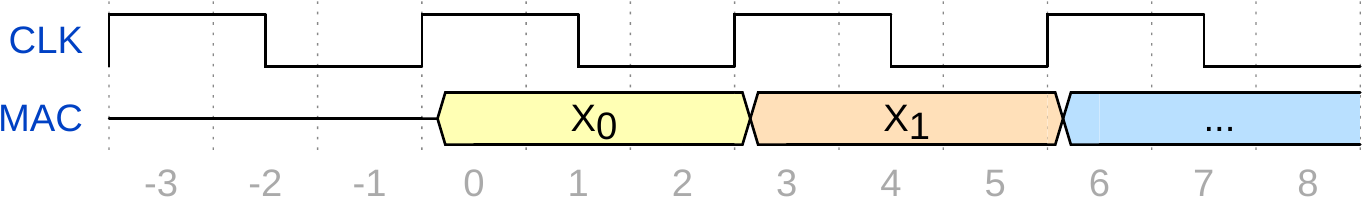}
        \caption{LUT-based state generation.} 
    \end{subfigure}
    \smallskip
    \begin{subfigure}[t]{0.45\textwidth}
        \centering
         \includegraphics[width=1\linewidth]{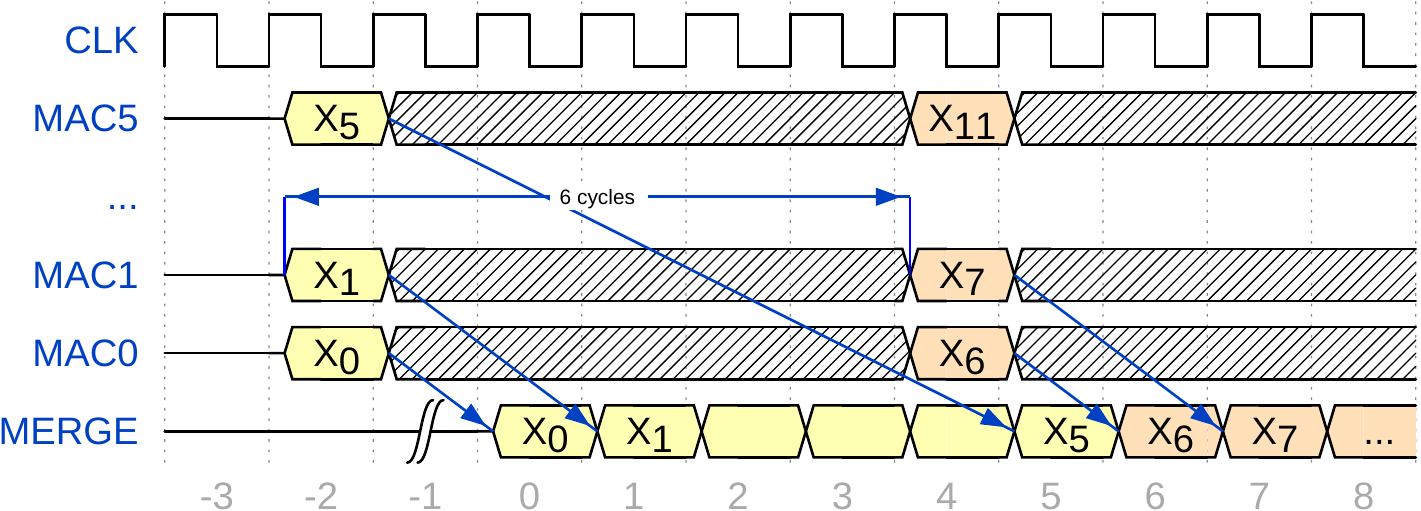}
        \caption{Proposed state generation with advance-6 recurrence.}
    \end{subfigure}
    \caption{The timing diagrams of different hardware designs for root state generation unit.}
    \label{fig:stu}
    
\end{figure}

\subsection{Sequence Output Unit (SOU)} \label{subsec:ppu}
{
With the root state as input, each SOU performs leaf state generation (LSGU), permutation, and decorrelation to finally output a random number sequence.
Each LSGU consists of an integer adder. The LSGU adder in the $i$-th SOU performs the addition of the root state with a unique constant value $h_i$, which is calculated at compile time by the approach described in~\Cref{subsec:statesharing}. 
The permutation unit is implemented by shift registers.
The rotation operation in the permutation function is divided into three stages to reduce the length of the combinatorial logic path. 
In the first stage, it calculates the the number of required bits for the rotation with the output from LSGU as input. 
In the second stage, it splits the rotation operation into two small rotations.
In the remaining stages, it performs these split rotations in parallel.
These stages are executed in a pipelined manner to guarantee a throughput of one output per cycle.
The decorrelator, which is a xorshift sequence generator and belongs to the linear feedback shift register generators, is implemented by the shift registers on FPGAs by following these previous works~\cite{panda2012fpga,bakiri2017ciprng}.
}

When increasing the number of SOUs to provide a massive number of sequences, state sharing by simple data duplication may cause a high fan-out problem since all SOUs require the same input from the RSGU.
This problem can be optimized by handcrafted register replication, but it loses the flexibility associated with HLS tools~\cite{de2018transformations}.
Therefore, we adopt a daisy chain scheme~\cite{boutros2018you} for the internal data transfer that each SOU receives data from the front SOU and then directly forwards the received data to the next SOU.
As there is no 1-to-N connection, it can keep the fan-out at a very low level at the cost of a slight increase in output latency. The extra latency is equal to the number of SOUs in the same topology times the period of the execution clock, which is only $1.82 \mu s$ for 1000 SOUs running at a frequency of 550MHz.


\subsection{Discussion}

{
Although ThundeRiNG is specifically designed for FPGAs, we also explore the possibility of generalizing our decorrelation and state sharing methods to CPUs and GPUs with the results presented in Section~\ref{subsubsec:compatility}. 
As this generalization is not the main focus of this paper, the implementations on CPUs and GPUs are rather straightforward, and we believe that more optimizations can be considered as future work.
For CPU implementation, we utilize a single thread for root state generation and multiple threads for parallel sequence output. The root states are generated in a batch manner so that the root states of each batch can fit in the last level cache for good data locality. In addition, we explore a double buffering scheme to overlap the root state generation and sequence output processes to increase throughput. 
For the GPU implementation, the state sharing mechanism leverages the shared memory hierarchy and hardware-accelerated arrive/wait barrier~\cite{cuRANDguide}. In one stream processor of GPU, we use a single thread for root state generation and multiple threads for parallel sequence output. The synchronization among them is managed through the efficient cooperative group interface provided by CUDA~\cite{cuRANDguide}.
}


\section{Evaluation}\label{sec:eva}

In this section, we evaluate both the quality and the throughput of ThundeRiNG.

\subsection{Experimental Setup}

\subsubsection{Hardware Platform}

The evaluation of statist`ical quality is conducted on a server with an Intel Xeon 6248R CPU and 768 GB DDR4 memory.
The throughput benchmarks and case studies are conducted on the following hardware platforms and corresponding development environments:

\begin{description}

\item[FPGA:] Xilinx Alveo U250 accelerator card with Vitis HLS Toolchain 2020.1. The number of available hardware resources are 2,000 BRAMs, 11,508 DSP slices, and 1,341,000 LUTs.
\item[GPU:] NVIDIA Tesla P100 GPU with CUDA Toolkit 10.1.
\item[CPU:] Two Intel Xeon 6248R CPUs (96 cores after hyperthreading enabled) with oneAPI Math Kernel Library 2021.2.   

\end{description}

\subsubsection{Parameter setting}\label{subsec:parameters}

The parameters of the root state transition include the modulus $m$, multiplier $a$ and increment $c$.
According to the existing empirical evaluation~\cite{l1999tables,o2014pcg}, we choose modulus $m$ as $2^{64}$, multiplier $a$ as $6364136223846793005$  and increment $c$ as $54$.
{
To guarantee scalability, we choose the xorshift128 generator as the decorrelator since it has the period of $2^{128} - 1$ and hence can generate $2^{64}$ nonoverlapping subsequences which satisfies the decorrelation requirement on $2^{63}$ distinct sequences~\cite{o2014pcg} of LCG with the state size of 64-bit.}
With the above parameters, ThundeRiNG is able to generate up to $2^{63}$ uncorrelated sequences, and the period of each sequence is up to $2^{64} - 1$.

\subsubsection{Evaluation Methods for Statistical Quality} 

To our best knowledge, there is no systematic benchmark for MISRN. Thus, we evaluate the quality of the MISRN generated by ThundeRiNG with two kinds of tests: intra-stream correlation and inter-stream correlation.
The intra-stream correlation indicates the dependence of random numbers from the same sequence (stream), while the inter-stream correlation indicates the dependence from different sequences.

\noindent
\textbf{Evaluation method on intra-stream correlation.}
Following previous studies~\cite{salmon2011parallel,o2014pcg}, we adopt a complete and stringent test suite, TestU01 suite~\cite{l2007testu01}, as the empirical test suite for statistical quality measurement. 
While existing works~\cite{li2011software,bakiri2018survey} only conducted tests with the Crush battery, we evaluate ThundeRiNG with the BigCrush battery, which is more extensive and has 64 more tests than the Crush battery~\cite{l2007testu01}.
{
Despite the BigCrush battery testing approximately $2^{38}$ random samples from one sequence, it can still miss regular patterns with a long period. 
We hence adopt a complimentary  test suite, PractRand~\cite{dotypractrand}, which allows for an unlimited test length of one sequence.}
PractRand runs in iterations. In each iteration, all tests are run at a given sample size. In the next iteration, the sample size is doubled until unacceptable failure occurs. Therefore, it is powerful to detect regular long-range patterns.
As ThundeRiNG can generate up to $2^{63}$ distinct sequences, it is impractical to evaluate them all. Hence, we randomly select 64 distinct sequences for evaluations.

\noindent
\textbf{Evaluation method on inter-stream correlation.}
{
As TestU01 and PractRand test suits are not designed for testing the inter-stream correlation~\cite{ismay2013testing}, we adopt the evaluation method from Li et al.~\cite{li2011software} that interleaves multiple sequences into one single sequence before evaluating the interleaved sequence with the BigCrush and PractRand test suites.
Specifically, the interleaved sequence is generated by selecting numbers from multiple sequences in a round-robin manner. 
Suppose there are $k$ sequences in total and the $i$-th sequence is $\set{x^{i}_{0}, x^{i}_{1}, ..., x^{i}_{n},...}$, the interleaved sequence will be $\set{x^{0}_{0},x^{1}_{0},...,x^{k}_{0},x^{0}_{1},x^{1}_{1},...x^{k}_{1},...}$. Besides TestU01 and PractRand, we also perform the Hamming weight dependency (HWD) test on the interleaved sequences. } 
{HWD, which is the dependency between the number of zeroes and ones in consecutive outputs, has been an important indicator of randomness and adopted by many test suites~\cite{l2007testu01}.
We use a powerful HWD testbench from Blackman et al.~\cite{blackman2018scrambled} that several existing crush-resistant PRNGs fail to pass~\cite{blackman2018scrambled}. }

{
Beyond this commonly adopted evaluation, to demonstrate the strength of our decorrelation method, we conduct a more stringent analysis on inter-stream correlation by three pairwise correlation evaluations. The experiments on pairwise correlation include Pearson’s correlation coefficient, Spearman’s rank correlation coefficient, and Kendall’s rank correlation coefficient~\cite{taylor1990interpretation}.
}

\begin{itemize}
  \item The Pearson correlation is also known as cross-correlation, which measures the strength of the linear relationship between two sequences.
  \item Spearman’s rank correlation represents the strength and direction of the monotonic relationship between two variables and is commonly used when the assumptions of Pearson correlation are markedly violated.
  \item Kendall rank correlation describes the similarity of the ordering of the data when sorted by each of the quantities.
\end{itemize} 
The outcomes of the three pairwise correlation tests range from -1 to +1, where -1 indicates a strong negative correlation, +1 for a strong positive correlation, and 0 for independence. As it is hard to traverse and analyze the pairwise correlations of all candidate sequences, we randomly select a pair of distinct sequences to calculate their coefficients and report the maximal correlation for 1000 such pairs.



\subsubsection{Methods for Throughput Evaluation}
{
We first evaluate the throughput of ThundeRiNG by varying the number of PRNG instances, and then comparing it with the state-of-the-art FPGA-based designs as well as CPU/GPU designs. 
For each experiment, we repeat the execution for 10 times and report the median throughput. 
}

{
There are, in general, two performance metrics for PRNG throughput evaluation: terabits generated per second (Tb/s), and giga samples generated per second (GSample/s).
Tb/s is commonly used in FPGA-based evaluation since FPGA-based PRNGs tend to have a large and arbitrary number of output bits (e.g., LUT-SR~\cite{thomas2012lut} uses 624-bit output) to increase the throughput.
GSample/s is used in CPU/GPU-based evaluation, where the size of a sample is usually aligned with 32-bit.
Hence, we use Tb/s when comparing with other FPGA-based implementations, and
GSample/s with the sample size of 32-bit when comparing with CPU/GPU-based implementations. 
For implementations with a larger sample size, we normalize it to 32-bit correspondingly. For example, Philox-4×32 uses the 128-bit round key~\cite{salmon2011parallel}, which produces $4 \times 32$-bit random numbers per output. We will count that as four samples per output, for a fair comparison. 
}


\subsection{Quality Evaluation}

\subsubsection{TestU01 and PractRand}
~\Cref{tab:overall_statistic_results} shows the testing results of ThundeRiNG and the state-of-the-art PRNG algorithms~\cite{blackman2018scrambled,salmon2011parallel,o2014pcg,l1999good,thomas2012lut} on BigCrush testing battery and PractRand test suite.

The results indicate that ThundeRiNG passes all tests in the BigCrush battery for both single sequence and multiple sequences.
The results of the PractRand suite show ThundeRiNG never encounters a defect even after outputting up to 8 terabytes random numbers. In summary, ThundeRiNG demonstrates a competitive quality of statistical randomness compared to the state-of-the-art PRNG algorithms.


\begin{table}[h!]

\caption{Statistical testing of ThundeRiNG and state-of-the-art PRNG algorithms on BigCrush and PractRand test suites.}
\label{tab:overall_statistic_results}
\resizebox{0.486\textwidth}{!}{%
\begin{tabularx}{0.61\textwidth}{lbcbc}
\toprule

\multirow{2}{*}{\textbf{Algorithms}}
&\multicolumn{2}{c}{\multirow{1}{*}{\textbf{Intra-stream correlation}}}
&\multicolumn{2}{c}{\multirow{1}{*}{\textbf{Inter-stream correlation}}} \\ [0.1em]\cline{2-5}

\\[-1em]
& \multicolumn{1}{c}{\textbf{~BigCrush~}}&\textbf{~PractRand~}&\multicolumn{1}{c}{\textbf{~BigCrush~}} & \textbf{~PractRand~}\\
\midrule

Xoroshiro128**~\cite{blackman2018scrambled}&\multicolumn{1}{c}{Pass} & >8TB & \multicolumn{1}{c}{Pass} & >8TB \\

Philox4\_32~\cite{salmon2011parallel}& \multicolumn{1}{c}{Pass} & >8TB & \multicolumn{1}{c}{Pass} & 1TB \\

PCG\_XSH\_RS\_64~\cite{o2014pcg}& \multicolumn{1}{c}{Pass} & >8TB & \multicolumn{1}{c}{105 failures} & 256MB \\

MRG32k3a~\cite{l1999good} & \multicolumn{1}{c}{Pass} & >8TB & \multicolumn{1}{c}{1 failure} & 2TB \\

LUT-SR~\cite{thomas2012lut} & \multicolumn{1}{c}{2 failures} & >1TB & \multicolumn{1}{c}{Pass} & 16MB \\

\textbf{ThundeRiNG}& \multicolumn{1}{c}{\textbf{Pass}} & \textbf{>8TB} & \multicolumn{1}{c}{\textbf{Pass}} & \textbf{>8TB}  \\

\bottomrule

%
%

\end{tabularx}%
}

\end{table}

\subsubsection{Pairwise Correlation Evaluation} \label{subsubsec:correlation_evalutation}

~\Cref{tab:correlation_results} shows the evaluation of pairwise correlation analysis when we enable different techniques: original LCG, original LCG + decorrelation, original LCG + permutation, and ThundeRiNG. The results indicate that the three kinds of correlations of multiple sequences generated by ThundeRiNG are much smaller than those by other design solutions, demonstrating the good statistical randomness of ThundeRiNG.

\begin{table}[h!]
\vspace{-2mm}
\caption{Pairwise correlation tests with different techniqes enabled.}
\label{tab:correlation_results}
\resizebox{0.476\textwidth}{!}{%
\begin{tabularx}{0.64\textwidth}{lcccc}

\toprule
\multirow{1}{*}{\textbf{Inter-stream }}
& \multirow{1}{*}{\textbf{LCG}} & \multicolumn{1}{c}{\textbf{LCG + }}  & \multicolumn{1}{c}{\textbf{LCG +}}   & \multirow{2}{*}{\textbf{ThundeRiNG}}
\\

\multirow{1}{*}{\textbf{Correlations}}
& \multirow{1}{*}{\textbf{Baseline}} & \multirow{1}{*}{\textbf{Decorrelation}}  & \multirow{1}{*}{\textbf{Permutation}} &
\\

\midrule
\multirow{1}{*}{{ Pearson }}
& {0.99764} & 0.00151 & 0.00019 &  \textbf{0.00003}
\\

\multirow{1}{*}{{ Spearman’s rank }}
& {0.99764} & 0.00150 & 0.00020 &  \textbf{0.00003}
\\
\multirow{1}{*}{{ Kendall’s rank }}
& {0.99843} & 0.00101 & 0.00013 &  \textbf{0.00002}
\\

\bottomrule

%

\end{tabularx}%
}

\end{table}

\subsubsection{Hamming Weight Dependency Evaluation}
~\Cref{tab:correlation_results_b} shows the evaluation results of Hamming weight dependency of different methods. The value of a result of the HWD test indicates the number of generated random numbers before an unexpected pattern is detected. Thus, a higher value of the result means better statistical quality.

We examine the impact of different techniques. If we only apply the permutation to the original LCG, there is no reduction in Hamming weight dependency, {although it reduces the collinearity significantly (shown in~\Cref{tab:correlation_results})}. In contrast, our decorrelation method significantly reduces the Hamming weight dependency.


\begin{table}[h!]
\vspace{-2mm}
\caption{Hamming weight dependency test with different techniqes enabled.}
\label{tab:correlation_results_b}
\resizebox{0.476\textwidth}{!}{%
\begin{tabularx}{0.64\textwidth}{lcccc}

\toprule
\multirow{1}{*}{\textbf{Inter-stream }}
& \multirow{1}{*}{\textbf{LCG}} & \multicolumn{1}{c}{\textbf{LCG + }}  & \multicolumn{1}{c}{\textbf{LCG +}}   & \multirow{2}{*}{\textbf{ThundeRiNG}}
\\

\multirow{1}{*}{\textbf{Correlations}}
& \multirow{1}{*}{\textbf{Baseline}} & \multirow{1}{*}{\textbf{Decorrelation}}  & \multirow{1}{*}{\textbf{Permutation}} &
\\

\midrule

\multirow{1}{*}{{ Blackman et al.~\cite{blackman2018scrambled}}}
&$1.25e+08 $ & $ > 1e+14$  & $1.25e+08 $  &  $\mathbf{> 1e+14}$
\\

\bottomrule

%

\end{tabularx}%
}

\end{table}

\begin{figure}[ht!]
     \centering
     \includegraphics[width=0.92\columnwidth]{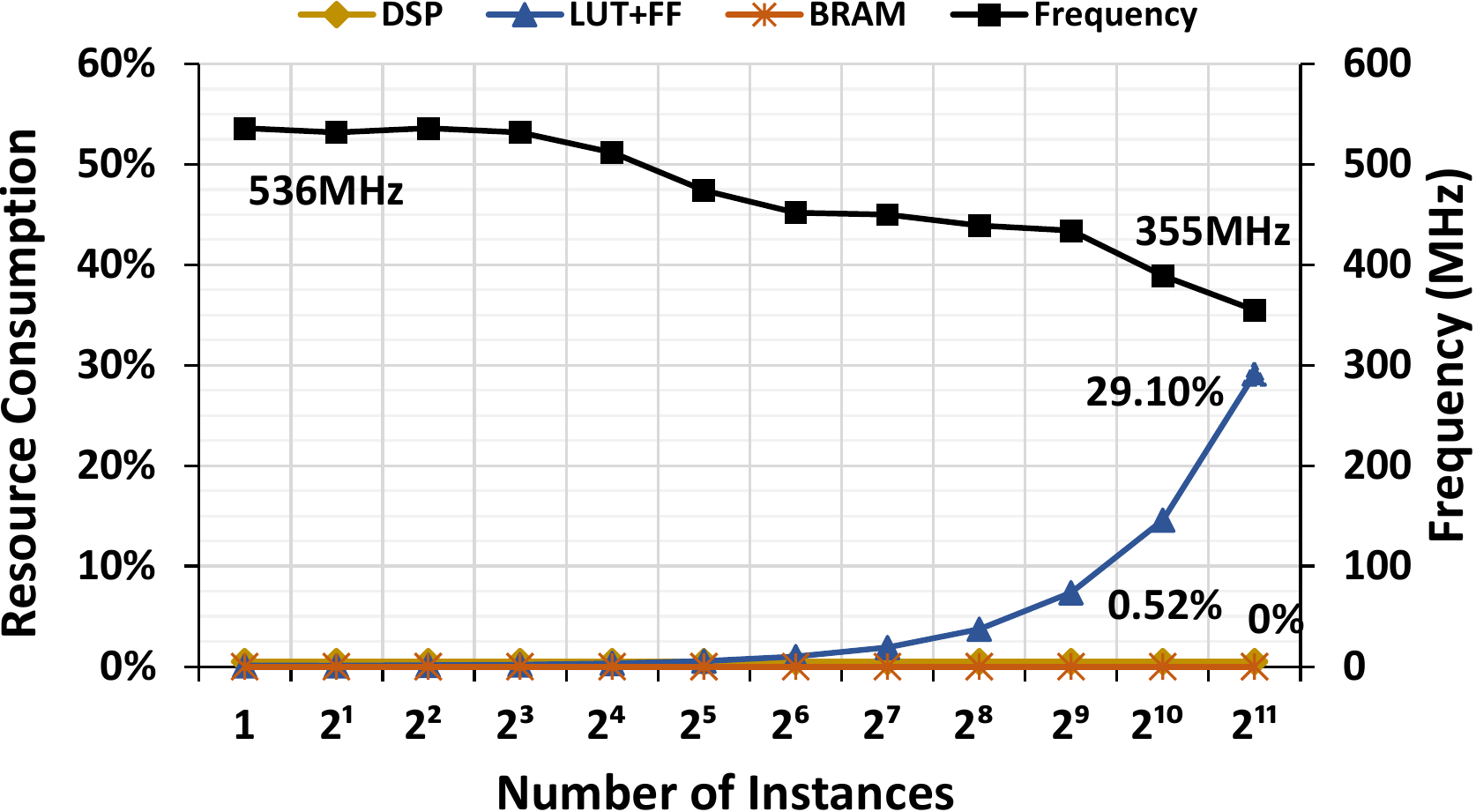}
     \caption{Resources consumption and clock frequency with varying number of SOU instances.}
     \label{fig:resource_by_instance}

\end{figure}

\begin{figure}[ht!]
     \centering
     \includegraphics[width=0.90\columnwidth]{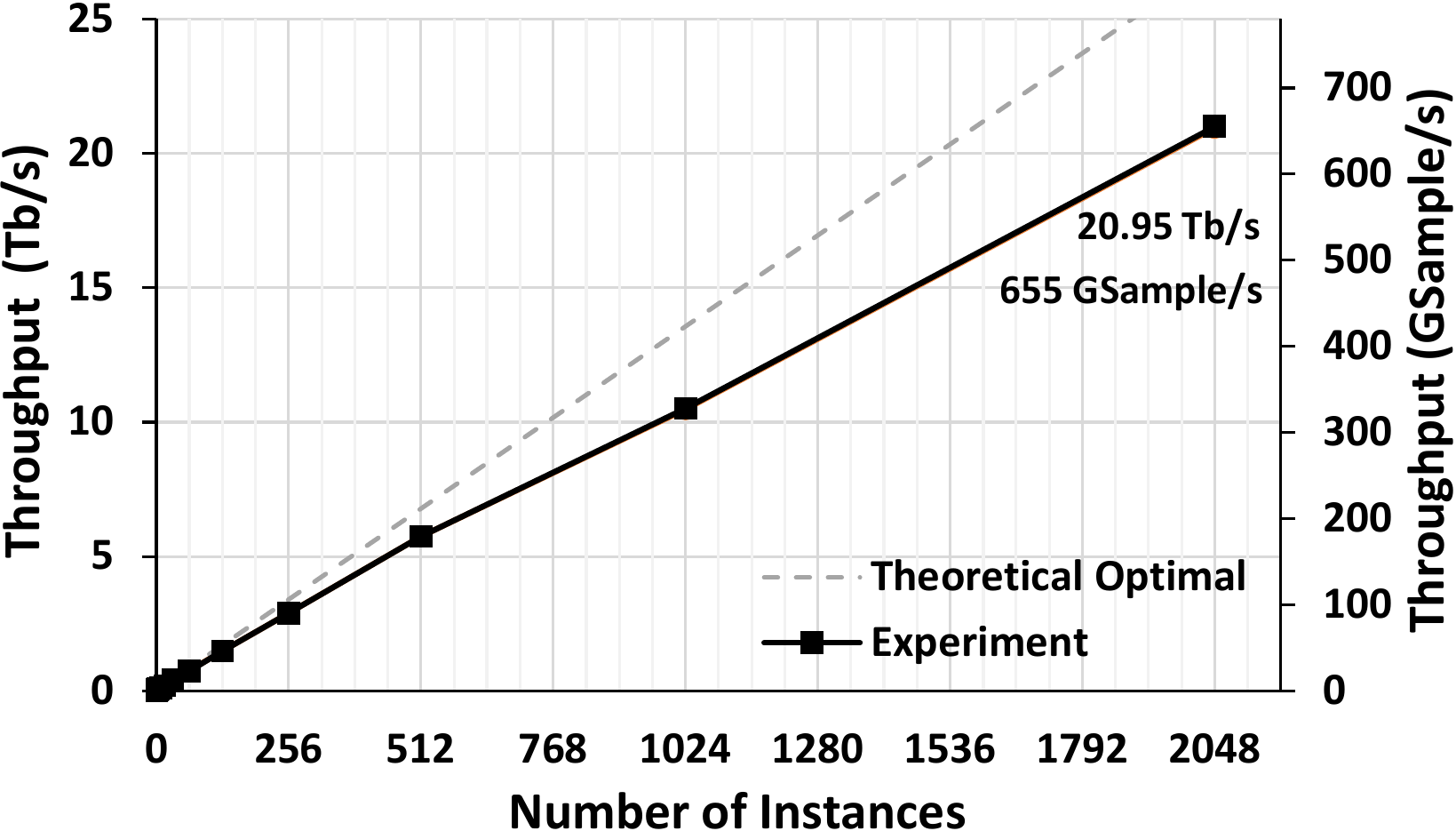}
     \caption{Throughput with varying number of SOU instances.}
     \label{fig:throughput_by_instance}

\end{figure}

\subsection{Throughput Evaluation on FPGA} \label{subsubsec:throughput_eva}

We evaluate the throughput and resource consumption of the proposed FPGA implementation of ThundeRiNG.
\Cref{fig:resource_by_instance} shows the resource consumption and the implementation frequency with the increasing number of PRNG instances.
The results show that DSP consumption is less than 1\% of the total capacity and, more importantly, it is oblivious to the number of instances.
This is because only the root state generation unit that requires the multiplication operation consumes the DSP resource, and ThundeRiNG only needs one root state generation unit for MISRN generation.
In addition, ThundeRiNG does not occupy any BRAM resource since the state is small enough to fit into the registers, and thus the BRAM utilization is 0\%.
The frequency is up to 500MHz and gradually drops as increasing number of instances due to more resource (FF + LUT) consumption.

\Cref{fig:throughput_by_instance} shows the overall throughput with increasing number of instances, where the solid black line is the measured throughput, and the dashed grey line is the optimal throughput under the frequency of 550MHz. 
The observed throughput is nearly proportional to the number of instances and can be up to 20.95 Tb/s with 2048 instances.
The gap between the optimal line and our results is because of the frequency drop (from 536MHz to 355MHz).

\subsubsection{Comparison with FPGA-based Works} \label{subsubsec:fpga_compare}
\Cref{tab:capacity} shows the performance comparison between ThundeRiNG and other state-of-the-art FPGA works as well as implementations with optimistic scaling. We estimate the throughput of the FPGA methods~\cite{li2011software, thomas2012lut} with optimistic scaling, where we assume the number of PRNG instances scales perfectly within the resource capacity, and the implementation frequency is fixed at 500MHz. In addition, we estimate the performance of porting high-quality CPU-based solutions~\cite{salmon2011parallel,blackman2018scrambled} to run on FPGAs. 
The estimated number of PRNG instances of CPU-based implementation on FPGAs is equal to the resource capacity of the FPGA platform divided by the resource consumption of one PRNG reported by the synthesis of the Vitis tool. We assume they have the 500MHz frequency on FPGA.


\begin{table}[t]

\caption{Throughput, quality test and resource utilization of the state-of-the-art FPGA-based works and porting CPU-based designs to FPGAs.}
\label{tab:capacity}
\resizebox{0.476\textwidth}{!}{%
\begin{tabularx}{0.67\textwidth}{llccccXc}
\toprule
\multicolumn{1}{l}{\multirow{2}{*}{\textbf{PRNGs}}}
& \multirow{2}{*}{\textbf{Quality}}
& \multirow{1}{*}{\textbf{Freq.}}
& \multirow{1}{*}{\textbf{Max}}
& \multirow{1}{*}{\textbf{BRAM}}
& \multirow{1}{*}{\textbf{DSP}}
& \multicolumn{1}{c}{\multirow{1}{*}{\textbf{Thr.}}}
& \multirow{2}{*}{\textbf{Sp.}}

\\
&
&\multirow{1}{*}{\textbf{(MHz)}}
&\multicolumn{1}{c}{\multirow{1}{*}{\textbf{\#ins.}}}
&\multirow{1}{*}{\textbf{(\%)}}
&\multirow{1}{*}{\textbf{(\%)}}
&\multirow{1}{*}{\textbf{(Tb/s)}}
&
\\ \midrule
\multicolumn{2}{l}{\textbf{\textit{Implementation Benchmarking:}}} \\
\textbf{ThundeRiNG}              & \textbf{Crush-resistant}
&  \textbf{355} &  \textbf{2048}  &   \textbf{0\%}  &   \textbf{0.5\%}   & \textbf{20.95}   & \textbf{$1\times$}  \\


Li et al.~\cite{li2011software}      & Crushable
& 475 & 16  & 1.6\% &  0\%   & 0.24  & $87.08\times$   \\

LUT-SR~\cite{thomas2012lut} & Crushable 
& 600 & 1   & 0\%   & 0\%    & 0.37  & $55.9\times$ \\

\midrule

\multicolumn{2}{l}{\textbf{\textit{Optimistic Scaling:}}} \\

Philox4\_32~\cite{salmon2011parallel} &  Crush-resistant
& 500 & 442  & 0\% &  100\%  & 2.83   & $7.39\times$    \\

Xoroshiro128**~\cite{blackman2018scrambled}  &  Crush-resistant
& 500 & 1150  &  0\%  & 100\% &  18.40 & $1.14\times$   \\

Li et al.~\cite{li2011software}     & Crushable
& 500 & 1000  & 100\% &  0\%  & 16.00  & $1.37\times$    \\

\bottomrule

\end{tabularx}%
}
\end{table}

ThundeRiNG outperforms all other designs significantly, delivering 87.08$\times$ and 55.9$\times$ speedup over the state-of-the-art FPGA-based solutions~\cite{li2011software, thomas2012lut} while guaranteeing a high quality of randomness.
{More importantly, while Li et al.~\cite{li2011software} achieve a throughput of 16Tb/s with optimistic scaling, ThundeRiNG still has 37\% higher throughput. It is also noteworthy that ThundeRiNG consumes no BRAMs while they use up all BRAMs.
}
Even assuming that Philox4\_32~\cite{salmon2011parallel} and xoroshiro128**~\cite{blackman2018scrambled} are ideally ported to the FPGA platform, ThundeRiNG still delivers 7.39$\times$ and 1.15$\times$ speedups over them, respectively, with much lower resource consumption.


\begin{table}[ht!]
\vspace{-2mm}
\caption{Throughput of various GPU PRNG schemes running on Nvidia Tesla P100 compared to ThundeRiNG's throughput.}
\label{tab:gpu_throughput}
\resizebox{0.476\textwidth}{!}{%
\begin{tabularx}{0.6\textwidth}{lYYYY}
\toprule
\multicolumn{1}{c}{\multirow{1}{*}{\textbf{Algorithms}}}
& \multirow{2}{*}{\textbf{BigCrush}}
&\multicolumn{2}{c}{\multirow{1}{*}{\textbf{Throughput}}}
&\multicolumn{1}{c}{\textbf{ThundeRiNG's }} \\

\multicolumn{1}{c}{\multirow{1}{*}{\textbf{(cuRAND)}}} &&
\multicolumn{1}{c}{\textbf{GSample/s}} &
\multicolumn{1}{c}{\textbf{Tb/s}}      & \multicolumn{1}{c}{\textbf{Speedup}}  \\


\midrule
Philox-4×32~\cite{salmon2011parallel}          & Pass 
&   \multicolumn{1}{c}{61.6234 }& \multicolumn{1}{c}{1.9719}
&   $\bm{10.62\times}$     \\
MT19937~\cite{matsumoto1998mersenne}           & Pass 
&   \multicolumn{1}{c}{51.7373 }& \multicolumn{1}{c}{1.6556}
&   $12.65\times$          \\
MRG32k3a~\cite{l1996combined}                  & 1 failure 
&   \multicolumn{1}{c}{26.2662 }& \multicolumn{1}{c}{0.8405}
&   $\bm{24.92\times}$     \\
xorwow~\cite{marsaglia2003xorshift}            & 1 failure 
&   \multicolumn{1}{c}{56.6053 }& \multicolumn{1}{c}{1.8114}
&   $11.56\times$          \\
MTGP32~\cite{DBLP:journals/corr/abs-1005-4973} & 1 failure 
&   \multicolumn{1}{c}{29.1273 }& \multicolumn{1}{c}{0.9321}
&   $22.47\times$          \\
\bottomrule

\end{tabularx}%
}
\end{table}

\subsection{Comparison with Existing Works on GPUs}
We perform throughput and quality comparison with GPU-based PRNGs in cuRAND~\cite{cuRANDguide}, which is the official library from Nvidia, as shown in ~\Cref{tab:gpu_throughput}. The statistical test results of cuRAND on the BigCrush battery are collected from the official document~\cite{cuRANDguide}.
The results show the ThundeRiNG outperforms GPU-based solutions from 10.6$\times$ to 24.92$\times$. 
On the other hand, three of the GPU-based PRNGs fail to pass the BigCrush test, while ThundeRiNG passes all tests. These experiments indicate that ThundeRiNG outperforms cuRAND in both throughput and quality.

\subsection{ThundeRiNG on CPU and GPU} \label{subsubsec:compatility}

{As a sanity check, we evaluate the design of ThundeRiNG on the CPU/GPU. 
\Cref{fig:thundering_on_cpu} compares the throughput of porting the design of ThundeRiNG to CPU/GPU with the state-of-the-art CPU/GPU-based PRNG implementations (Intel MKL and cuRAND).
ThundeRiNG did not perform well on the CPU when the number of instances is larger than $2^4$ because the overhead of CPU synchronization for state sharing rises dramatically. 
ThundeRiNG on GPU slightly outperforms cuRAND.  To be fair, more optimizations for these implementations are needed in the future. For example, the ThundeRiNG design on the CPU can be improved by utilizing SIMD, and ways can be found to reduce the synchronization overhead.
}

{
In summary, because of its fine-grained parallelism and synchronization, the state sharing and decorrelation proposed fits FPGAs better, allowing for the instantiation of many MISRN generators.
Consequently, the throughput of ThundeRiNG on FPGA scales linearly with the number of generators and outperforms other designs significantly, even though it runs at a much slower frequency.
}

\begin{figure}[t]
     \centering
     \includegraphics[width=0.92\columnwidth]{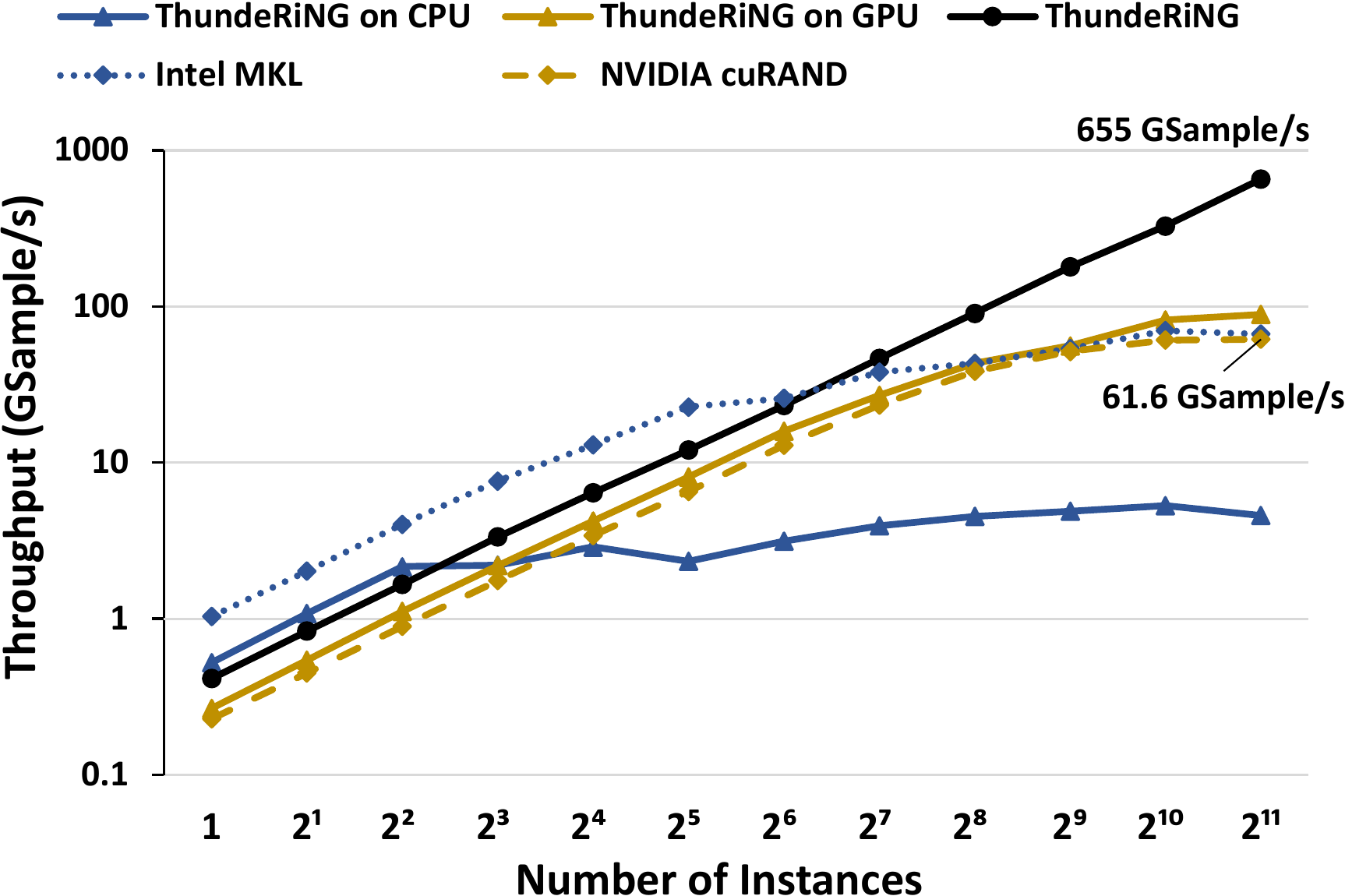}
     \caption{Performance comparison of the compatibility implementations of ThundeRiNG on CPU/GPU,  Intel MKL PRNG and NVIDIA cuRAND with varying number of instances.}
     \label{fig:thundering_on_cpu}
     \vspace{-3mm}
\end{figure}



\section{Case Studies}\label{sec:case}

To further demonstrate the advantages of ThundeRiNG, we apply it to two applications: the estimation of $\pi$ and Monte Carlo option pricing.
Furthermore, we compare the FPGA implementations with the GPU-based ones.

\subsection{Implementation}
Estimating the value of $\pi$ is widely used as an application to demonstrate the efficiency of the PRNG~\cite{li2011software,howes2007efficient}.
The basic idea is that assuming we have a circle and a square that encloses the circle, the value $\pi$ can be calculated by dividing the area of the circle by the area of the square.
In order to estimate the area of the circle and the square, we generate a large number of random points within the square and count how many of them falling in the enclosed circle.
Random number generation is the bottleneck of this application as it consumes 87\% of the total execution time of the GPU-based implementation according to our experiment.

Monte Carlo option pricing is commonly used in the mathematical finance domain. It calculates the values of options using multiple sources of random features, such as interest rates, stock prices, or exchange rates. It relies on PRNGs to generate a large number of possible but random price paths for the underlying of derivatives. The final decision is made by computing the associated exercise value of the option for each path.
We choose the Black-Scholes model as the target model for option pricing.
On the GPU-based implementation, random number generation accounts for 54\% of the total execution time, according to our experiment.

We implement two applications on both GPU and FPGA platforms for comparison.
For GPU-based implementations, we directly use the officially optimized designs of Nvidia using the cuRAND library~\cite{cuRANDguide}, targeting the Nvidia Tesla P100 GPU.
For our FPGA-based implementation, we use the ThundeRiNG for random number generation and design the rest of the logic in the application using HLS to achieve the same functionality as the GPU-based designs.
For all implementations, single-precision floating points were used as the data type.


\begin{figure}[t]
     \centering
     \includegraphics[width=0.92\columnwidth]{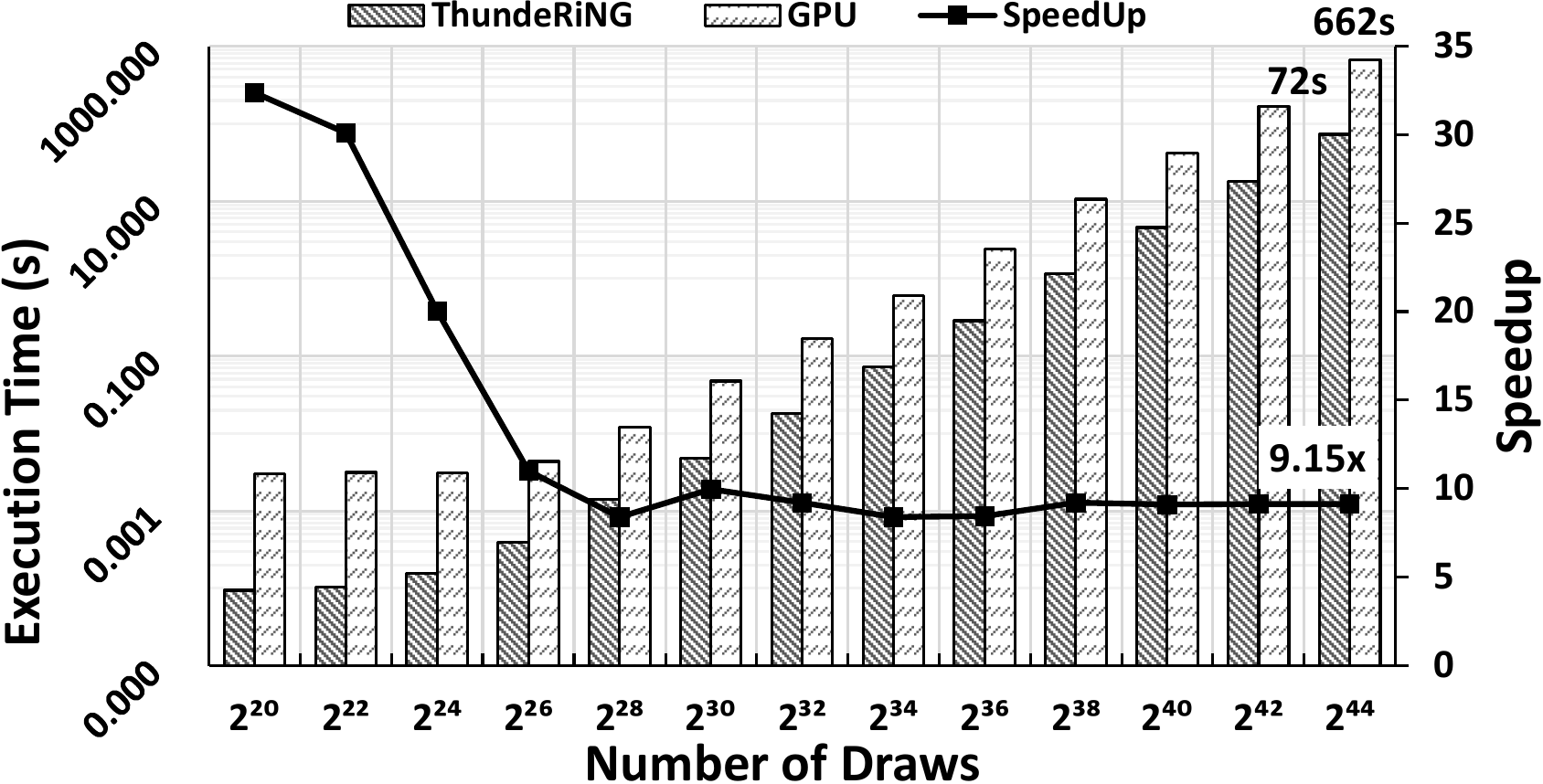}
     \caption{Execution time of estimation of $\bm{ \pi }$ of FPGA-based solution (ThundeRiNG) and GPU-based solution with varying number of draws.}
     \label{fig:case1}
\end{figure}

\begin{figure}[t]
     \centering
     \includegraphics[width=0.92\columnwidth]{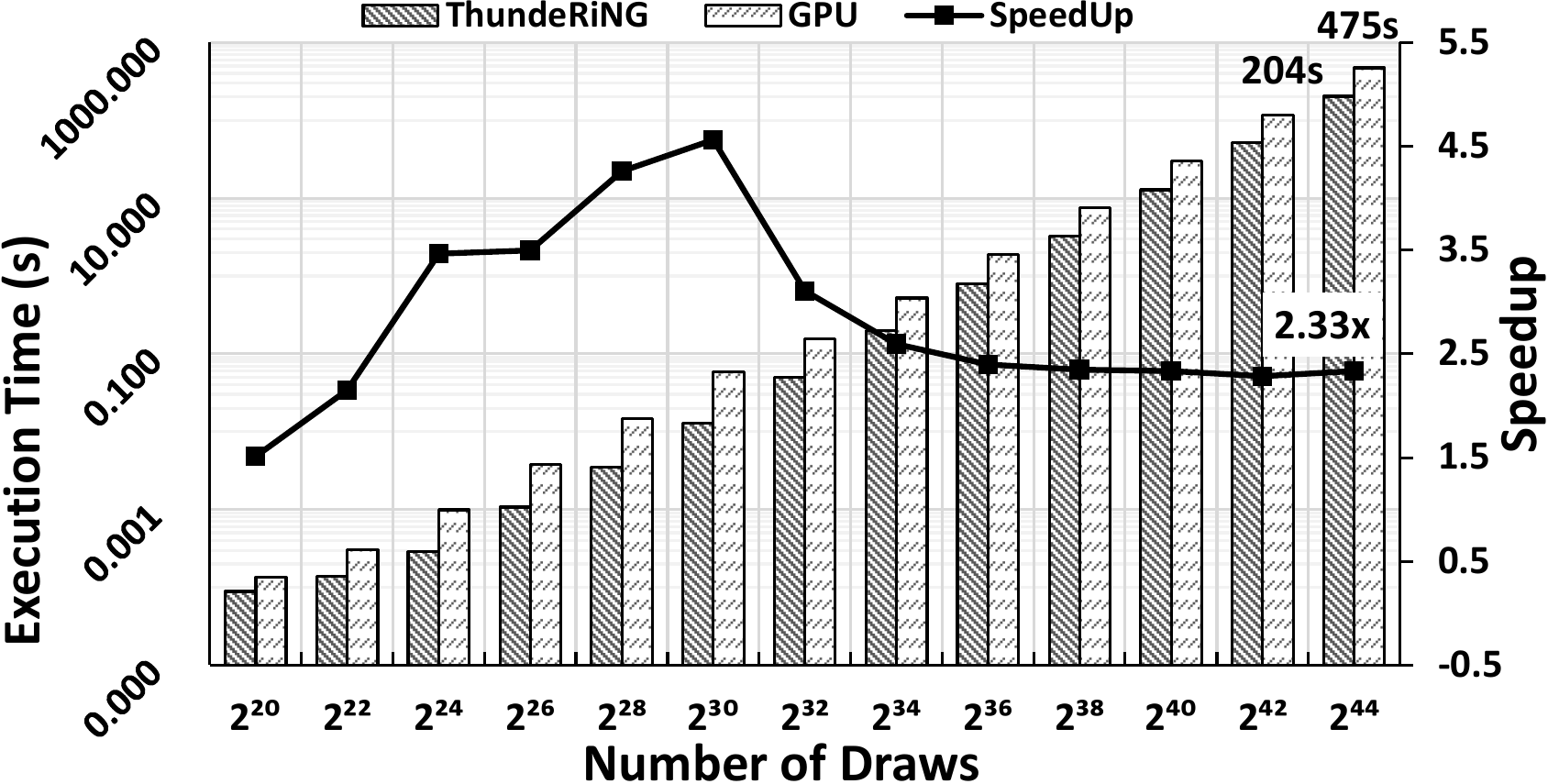}
     \caption{Execution time of Monte Carlo option pricing of FPGA-based solution (ThundeRiNG) and GPU-based solution with varying number of draws.}
     \label{fig:case2}
\end{figure}

\subsection{Results}

\Cref{fig:case1} shows the performance of FPGA-based solution (ThundeRiNG) and GPU-based solution on the estimation of $\bm{ \pi }$  with varying the number of draws, where each draw requires two random numbers. The results show FPGA-based solution significantly outperforms GPU-based solution for all number of draws, and the speedup is stable and up to 9.15$\times$ for the massive number of draws. The downgrade trend of speedup is because GPU-based implementation cannot fully utilize the hardware capacity when the number of draws is not large.

\Cref{fig:case2} shows the performance of FPGA-based solution (ThundeRiNG) and GPU-based solution on Monte Carlo option pricing with varying number of draws, each draw requiring a new random number. Our implementation with ThundeRiNG significantly outperforms the GPU-based solution for all number of draws. The speedup of the massive number of draws can be up to 2.33$\times$.

In addition to the comparison on throughput, we also show the resource utilization of the FPGA platform and the power efficiency comparison between GPU and FPGA, as shown in \Cref{tab:gpu_application}, where the power consumption is reported by respective official tools, namely {\tt nvidia-smi} for GPU and {\tt xbutil} for FPGA, and the power efficiency is calculated by dividing the throughput by the power consumption.
The results show the FPGA-based solutions outperform the GPU-based solutions by 6.83$\times$ and 26.63$\times$ for MC option pricing and the estimation of $\pi$, respectively. The end-to-end comparison of the two applications demonstrates that ThundeRiNG is able to generate massive independent random numbers with high throughput, and FPGA can be a promising platform for PRNG involved applications.


\begin{table}[t]
\caption{The comparison of throughput and power efficiency of two applications between FPGA and GPU.}
\label{tab:gpu_application}
\resizebox{0.476\textwidth}{!}{%

\begin{tabularx}{0.68\textwidth}{llcc}
\toprule

\multicolumn{2}{l}{\multirow{1}{*}{\textbf{Applications}}}
& \multirow{1}{*}{\textbf{Estimation of $\bm{ \pi }$}}
& \multirow{1}{*}{\textbf{MC option pricing}} \\\midrule


& Frequency (MHz) & 304      &  335 \\

& Number of instances  & 1600 &  256  \\

\multicolumn{1}{c}{\multirow{2}{*}{ FPGA: Alveo U250 }}
& LUTs    & 1048235(70\%)  &  735173(49\%)    \\

\multicolumn{1}{c}{\multirow{2}{*}{ (16nm FinFET) }}
& FFs     & 1171130(38\%)  &   751810(24\%)   \\
& DSPs    & 5512(45\%)     &   5984(49\%)      \\
&Throughput (GSample/s) & 480 & 86  \\
&Power consumption (W) & 45  & 43 \\ \midrule

\multicolumn{1}{c}{\multirow{2}{*}{GPU: Tesla P100 } }

& Frequency (MHz) & 1,190  &  1,190 \\

\multicolumn{1}{c}{\multirow{2}{*}{ (16nm FinFET) }}
&Throughput (GSample/s) & 53 & 33 \\


&Power consumption (W) & 131  & 126 \\  \midrule

\multicolumn{1}{c}{\multirow{1}{*}{ThundeRiNG's}}
& Throughput speedup& \textbf{9.15x}  &  \textbf{2.33x} \\

\multicolumn{1}{c}{\multirow{1}{*}{improvement}}
&Power efficiency &  \textbf{26.63x} & \textbf{6.83x} \\\bottomrule


%
%
%

\end{tabularx}%
}
 \vspace{-3mm}
\end{table}

\section{Conclusion}\label{sec:conclusion}
In this paper, we propose the first high-throughput FPGA-based crush-resistant PRNG called ThundeRiNG for generating multiple independent sequences of random numbers. Theoretical analysis shows that our decorrelation method can enable the concurrent generation of high-quality random numbers. By sharing the state, ThundeRiNG uses a constant number of multipliers and BRAM regardless of the number of sequences to be generated. Our results show that ThundeRiNG outperforms all current FPGA and GPU based pseudo-random number generators significantly in performance as well as quality of the output. Furthermore, ThundeRiNG is designed to be used as a `plug-and-play' IP block on FPGAs for developer convenience. We believe that our work contributes to making the FPGA a promising platform for high performance computing applications.

\begin{acks}
{
We thank the anonymous reviewers for their valuable feedback on this work.
We thank the Xilinx Adaptive Compute Cluster (XACC) Program~\cite{xacc} for the generous donation.
This work is supported by MoE AcRF Tier 1 grant (T1 251RES1824), Tier 2 grant (MOE2017-T2-1-122) in Singapore, and also partially supported by the National Research Foundation, Prime Minister's Office, Singapore under its Campus for Research Excellence and Technological Enterprise (CREATE) programme.
}
\end{acks}

\balance
\bibliographystyle{ACM-Reference-Format}
\bibliography{ref}

\end{document}